\documentclass[preprint,review,12pt,letterpaper]{elsarticle}
\usepackage{fullpage}
\usepackage{etex}
\usepackage{tabularx}  
\usepackage{algorithm,algpseudocode,graphicx} 
\usepackage{mathtools}
\usepackage{amsfonts}
\usepackage{tabularx} 
\usepackage{subfigure}
\usepackage{color}
\usepackage{amssymb}  
\usepackage{epsfig} 
\usepackage{enumitem}
\usepackage{amsmath}   
\usepackage[numbers]{natbib}
\newtheorem{theorem}{Theorem}
\newtheorem{remark}{Remark}

\newtheorem{proposition}{Proposition}
\newtheorem{corollary}{Corollary}
\newproof{proof}{Proof}

\journal{Physical Communication}

\begin{document}
	\begin{frontmatter}		
		\title{Subcarrier Pairing as Channel Gain Tailoring: Joint Resource Allocation for Relay-Assisted Secure OFDMA with Untrusted Users} 
		\tnotetext[t1]{A preliminary version of this work was presented at the IEEE GLOBECOM Wkshp, Trusted Commun., Physical layer security, Dec. 2017, Singapore \cite{GCW17_SCP}.}
		
		\author[iitj]{Ravikant Saini}\ead{ravikant.saini@iitjammu.ac.in}
		\author[liu]{Deepak Mishra}\ead{deepak.mishra@liu.se}
		\author[iitdee]{Swades De\corref{cor1}}
		\cortext[cor1]{Corresponding author. Tel.: +91.11.2659.1042; fax: +91.11.2658.1606.}\ead{swadesd@ee.iitd.ac.in}
		
		\address[iitj]{Department of Electrical Engineering, Indian Institute of Technology Jammu, Jammu 181221, India}
		\address[liu]{ Department of Electrical Engineering, Link\"oping University, Link\"oping 58183, Sweden}	
		\address[iitdee]{Department of Electrical Engineering, Indian Institute of Technology Delhi, New Delhi 110016, India}		
		 
		\begin{abstract}			
			Joint resource allocation involving optimization of subcarrier allocation, subcarrier pairing (SCP), and power allocation in a cooperative secure {\color{black} orthogonal frequency division multiple access} (OFDMA) communication system with untrusted users is considered. Both amplify and forward (AF), and decode and forward (DF) modes of operations are considered with individual power budget constraints for source and relay. After finding optimal subcarrier allocation for an AF relayed system, we prove the joint power allocation as a generalized convex problem, and solve it optimally. Compared to the conventional channel gain matching view, {\color{black} the} optimal SCP is emphasized as a novel concept of channel gain tailoring. We prove that {\color{black} the} optimal SCP pairs subcarriers such that the variance among the effective channel gains is minimized. For {\color{black} a} DF relayed system, we show that depending on {\color{black} the} power budgets of source and relay, SCP can  either be in a subordinate role where it improves the energy efficiency, or in a main role where it improves the spectral efficiency of the system. In {\color{black} an} AF relayed system we confirm that SCP plays a crucial role, and improves the spectral efficiency of the system. {\color{black} The} channel gain tailoring property of SCP, various roles of SCP in improving {\color{black} the} spectral and energy efficiency of a cooperative communication system are validated with the help of simulation results.
			
		\end{abstract}
		
	\end{frontmatter}
	 
	\section{Introduction}\label{sec_introduction}
	Mobile communication providing freedom from a tethered connection has led to an era of personalized world where users use network infrastructure at their will and comfort. With increasing usage of mobile applications for online transactions, the need of having strong security system to protect important information is increasing. In comparison to cryptographic techniques, physical layer security appears as a less-complex cost-effective solution \cite{amitav_TCST_2014}. For providing secure communication to {\color{black} all} users, physical layer security has recently gathered much interest in cooperative communication \cite{Bassily_SPM_2013}.
	
	The concept of information theoretic security was introduced in a landmark paper by Wyner \cite{Wyner1975}, where the author proved {\color{black} the}  feasibility of sending a message reliably to a destination, keeping it secret from an eavesdropper. Physical layer security utilizes {\color{black} the} inherent independence of subcarriers in an orthogonal frequency division multiple access (OFDMA) system. It has been investigated recently for OFDMA-based broadcast and cooperative communication for next generation communication systems, such as, fourth-generation (4G) LTE and {\color{black} fifth-generation} (5G) systems.  {\color{black}An exhaustive survey on recent cooperative relaying and jamming strategies for physical layer security has been provided in \cite{Jameel_CST_2018}.} Below, we survey the related works on secure OFDMA systems.
	 
	\subsection{Prior Art}
	\subsubsection{Resource Allocation without Subcarrier Pairing}
	In physical layer security, the system models are {\color{black} generally classified in two broad categories -} one having external eavesdropper which is an entity external to the network \cite{LLai_TIT_2008, Jeong_TSP_2011, Jindal_CL_2015, Deng_TIFS_2015, Hmwang_TC_2015, Derrick_TWC_2011, HMWang_TVT_2015}, and the other having untrusted users which are actually legitimate users of the system, {\color{black} who} have lost mutual trust and consider each other as {\color{black} an} eavesdropper  \cite{Jorswieck2008,Xiaowei_TIFS_2011,RSaini_CL_2016}.
Untrusted users is a far more hostile condition in comparison to external eavesdropper, {\color{black} where} each user contends for a subcarrier against rest of the users (behaving as eavesdroppers), which results in {\color{black} a} relatively complex resource allocation problem and lesser secure rate \cite{rsaini_TIFS_2016}. 
	
	{\em External Eavesdropper:} 
	In cooperative secure communication, several relaying strategies were proposed in \cite{LLai_TIT_2008}. In a {\color{black} decode and forward} (DF) relayed single source destination pair communication with direct link availability, and in presence of an external eavesdropper, power allocation was solved in \cite{Jeong_TSP_2011}. In a similar setup, resource allocation for an amplify and forward (AF) relayed cooperative communication was presented in \cite{Jindal_CL_2015}.  Utility of a helper node as a relay or a jammer in a similar four-node setup was discussed in \cite{Deng_TIFS_2015} for ergodic secrecy rate maximization. In a four-node setup with a trusted DF relay, outage constrained secrecy throughput maximization was investigated in \cite{Hmwang_TC_2015} considering system power budget and unavailability of direct path. With imperfect channel state information (CSI) at a multi-antenna source, resource allocation problem in a DF relay-assisted system in presence of a multi-antenna eavesdropper was considered in \cite{Derrick_TWC_2011}. The work in {\color{black}\cite{Abedi_TWC_2016} considered sum secure rate maximization for a multiple DF relay assisted secure communication system, constrained by limited feedback, in presence of multiple eavesdroppers.} Joint beamforming, jamming, and power allocation problem was considered in \cite{HMWang_TVT_2015} for a single source-destination pair communication assisted by multiple AF relays in presence of an external eavesdropper. {\color{black}A robust resource allocation framework to handle an active full-duplex eavesdropper has been considered in \cite{Abedi_TWC_2017} assuming a full-duplex receiver. With the availability of direct path, but without any information about eavesdropper's CSI, \cite{Poursajadi_TVT_2018} considered the problem of security as well as reliability in presence of multiple relays. Recently, \cite{Zhang_JSAC_2018} investigated subcarrier and power allocation in an AF relay assisted secure non-orthogonal multiple access (NOMA) communication system.} 
	
	{\em Untrusted users:}
Considering a broadcast OFDMA based secure communication system with two untrusted users, subcarrier and power allocation problem was investigated in \cite{Jorswieck2008}. The authors in \cite{Xiaowei_TIFS_2011} considered resource allocation problem for two classes of users: one having secure rate demands, and the other having best effort traffic. In {\color{black} \cite{rsaini_TIFS_2016}, the authors considered sum rate maximization and max-min fairness optimization problems in a  jammer assisted secure communication system with untrusted users.} Another important dimension to this field involves {\color{black} the} usage of helper nodes to improve secrecy performance of the communication system. In this direction, sum rate maximization and sum power minimization were presented in \cite{RSaini_CL_2016} for a DF relay-assisted system. With direct link availability in a DF relayed system, optimal transmission mode selection was investigated in \cite{RSaini_CL_2017}.
	
	\subsubsection{Resource Allocation with Subcarrier Pairing}
	Subcarriers on different hops in a cooperative communication system are independent. Mapping of subcarriers on two hops for performance optimization, is known as subcarrier pairing (SCP) \cite{Herdin_ICC_2006}. A few recent works that have explored SCP are discussed below.  
	
	{\em Non-secure OFDMA:} The authors in \cite{Hottinen_SPAWC_2007} proved that ordered pairing (OP) is optimal for an AF relayed network having no direct link. Subcarrier pairing in a non-secure OFDMA system with both AF and DF relay was presented in \cite{YLI_CL_2009}. Joint subcarrier pairing and power allocation in a DF relayed  communication for both joint and individual power budget constraints were studied in \cite{HSU_TSP_2011}. Joint resource allocation for two-way  AF relay-assisted multiuser networks was investigated in \cite{Guftar_TC_2011}.  
	
	{\em Secure OFDMA:} 
	Joint resource allocation problem involving subcarrier allocation, source power allocation and subcarrier pairing was studied in \cite{CCAI_WCSP_2015} for an AF relayed secure communication system considering equal power allocation at relay.  Energy-efficient resource allocation problem in a multi-user multi-relay scenario with an external eavesdropper was studied in \cite{ZChang_ICC_2016}. Joint resource allocation for secure OFDMA two-way relaying with {\color{black} an} external eavesdropper and multiple source-destination pairs was discussed in \cite{Zhang_TII_2016}.  
	
	\subsection{Research Gap and Motivation}
{\color{black}Even though secure cooperative communication has been extensively studied with external eavesdroppers, not many works have considered systems with untrusted users. \emph{To the best of our knowledge, AF relay assisted secure OFDMA communication system with untrusted users has not been investigated yet}. We consider an AF relay assisted system for its inherent simplicity in terms of design, implementation cost, easier maintenance and less stringent signal processing requirements. Since, the rate definitions in AF and DF relayed communication systems are different, these are to be investigated separately. Further, observing that the SCP investigated for an external eavesdropper is not usable for untrusted users, we investigate SCP for both AF and DF relayed system with untrusted users. \emph{We also believe that, joint resource allocation involving subcarrier allocation, power allocation, and subcarrier pairing has not been investigated for untrusted users.}}	
	Lastly, it is also worth noting that the prior works considering SCP in secure OFDMA have solved the problem in dual domain \cite{CCAI_WCSP_2015, ZChang_ICC_2016, Zhang_TII_2016}. Their dual decomposition based iterative solutions suffer from duality gap error which vanishes only with very large number of subcarriers. In contrast,  we investigate the problem in  \emph{primal domain} itself, and thus, the solution obtained does not suffer from duality gap. \textit{In particular, a near-global-optimal joint resource allocation is proposed  using an equivalent transformation {\color{black}for subcarrier allocation}, tight approximation {\color{black}for subcarrier pairing}, and generalized convexity principles {\color{black}for joint power allocation} \cite{Baz}; its quality has been both analytically described in Sections \ref{sec_sca_pa_af} and \ref{sec_opt_scp}, and numerically validated in Section \ref{sec_results}.}

	\subsection{Contribution and Scope of Work}
	The key contributions are summarized as follows:
	
	\begin{enumerate}[label=(\alph*)]
		\item {\color{black} We investigate a novel joint resource allocation problem for maximizing sum secure rate in both AF and DF relay-assisted secure OFDMA with untrusted users, and individual power budgets on source and relay. The practical constraints leads to a non-linear, non-convex optimization problem having  exponential complexity with number of subcarriers, which belongs to the class of NP hard problems.}  
		\item Combinatorial aspect of subcarrier allocation is solved optimally. Joint power allocation with individual power budgets is proved to be generalized convex, and solved optimally through {\color{black} Karush-Kuhn-Tucker} (KKT) conditions. Thereafter, the combinatorial issue of subcarrier pairing is solved near-optimally such that {\color{black} the} proposed SCP tends to be globally-optimal as signal-to-noise ratio (SNR) increases.   
		\item {\color{black} SCP as a novel system design concept, i.e., channel gain tailoring, and the spectral and energy efficiency improvements through optimal SCP as a steady step towards green communication are corroborated via extensive numerical results. Their efficacy is further strengthened by comparing the proposed schemes against the benchmarks.}   
	\end{enumerate}
	\begin{figure}[!t]
		\centering
		\epsfig{file=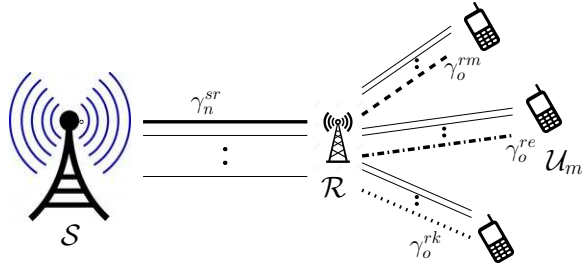,width=3in}\vspace{-0.5mm}
		\caption{System model describing the SCP $(n,o)$ between $\mathcal{S}-\mathcal{R}$ and $\mathcal{R}-\mathcal{U}$ links for different users.}
		\label{fig:system_model} 
	\end{figure}
	
	A preliminary version of this work was recently presented in \cite{GCW17_SCP}.
	
	Secure cooperative communication with untrusted users is a practical scenario. Communication to secondary users in presence of primary users in a cognitive radio networks, and communication to sensor nodes by using one of the nodes as  trusted relay in presence of nodes from other vendors in an {\color{black} internet of things} (IoT) based system are some of the {\color{black} example} scenarios with untrusted users, envisioned for advanced communication systems, such as 4G-LTE and 5G.
	The concepts of \textit{global-optimal subcarrier allocation, generalized convexity based power allocation, and channel gain tailoring based near optimal SCP}, as developed in this work, can be extended to secure communication systems with multiple relays/antennas, for investigating other quality of service (QoS) aware utilities, such as ergodic/outage capacity.

	\subsection{Paper Organization} 
	The system model and secure sum rate maximization problem are presented in Section \ref{sec_system_model}. The subcarrier and power allocation in AF and DF relay-assisted systems are considered in Section \ref{sec_sca_pa_af}. Optimal subcarrier pairing is discussed in Section \ref{sec_opt_scp}, and results are described in Section \ref{sec_results}. Concluding remarks are presented in Section \ref{sec_conclusion}.
	
\newcolumntype{C}{>{\centering\arraybackslash}p{3em}}
\newcolumntype{L}[1]{>{\raggedright\let\newline\\\arraybackslash}p{#1}}
\begin{table}[!htb]
\caption{Definitions of acronyms used}\label{acronym}
\centering
{\scriptsize
\begin{tabular}{|L{1.2cm} |L{3.3cm}| L{1.1cm}| L{3cm}| L{1.1cm}| L{4.1cm}|}
\hline
Acronym & Definition & Acronym & Definition & Acronym & Definition \\ 
\hline
OFDMA & Orthogonal division multiple access & KKT & Karush-Kuhn-Tucker condition & MINLP & Mixed integer non linear programming  \\ 
\hline 
DF & Decode and forward & AF & Amplify and forward & 4G & Fourth generation \\
\hline
LTE & Long term evolution & SNR & Signal to noise ratio & IoT & Internet of things\\
\hline
QoS & Quality of service & SCP & Subcarrier pairing & CSI & Channel state information \\
\hline
opt-SCP & Optimal SCP & OP & Ordered pairing & def-SCP & default SCP\\
\hline
OPA & Optimal power allocation & EPA & Equal power allocation & NOMA & Non-orthogonal multiple access\\
\hline
\end{tabular}
}
\end{table}

	\section{System Model and Problem Formulation}\label{sec_system_model} 
The downlink of a trusted relay assisted half duplex cooperative communication system with $N$ subcarriers and $M$ untrusted users $\mathcal{U}_m$ (cf. Fig. \ref{fig:system_model}) is considered. Relay $\mathcal{R}$ and source $\mathcal{S}$ have individual power budgets, which is more practical due to their geographically-separated locations. Secure communication with untrusted users is a practical scenario, with external eavesdropper's case being a simple extension. In this  hostile environment, users {\color{black} have} lost mutual trust, and request communication from $\mathcal{S}$ {\color{black} assuming} other users as eavesdroppers. 

	All nodes are {\color{black} assumed to be} equipped with single antenna \cite{Jeong_TSP_2011, Jindal_CL_2015}. Due to deep fading, users are not directly reachable to source \cite{Hmwang_TC_2015, HMWang_TVT_2015, Derrick_TWC_2011, RSaini_CL_2016}. To communication with users, $\mathcal{S}$ has to communicate through $\mathcal{R}$. All subcarriers over source to relay $\mathcal{S}-\mathcal{R}$ link, and relay to users $\mathcal{R}-\mathcal{U}$ link are considered to face quasi-static Rayleigh fading. Since all users are legitimate in the system, perfect CSI of all the links is  available at source using conventional channel estimation and feedback  mechanisms \cite{Jeong_TSP_2011, Jindal_CL_2015, RSaini_CL_2016, Jorswieck2008, Xiaowei_TIFS_2011}.
	
	\begin{remark}
		Resource allocation which includes subcarrier allocation, power allocation, and subcarrier pairing can be performed at $\mathcal{S}$ as well as $\mathcal{R}$. Assuming $\mathcal{S}$ to be a controlling node, $\mathcal{S}$ completes resource allocation before transmission of a frame, through cooperation from $\mathcal{R}$. 
	\end{remark}
	
	\subsection{Joint Resource Allocation Problem}
	A system with $M$ untrusted users is equivalent to a multiple eavesdropper scenario where each user has to contend against {\color{black}rest of the} $M-1$ potential eavesdroppers over each subcarrier. The strongest eavesdropper with maximum SNR is hereafter referred  to as  \textit{equivalent eavesdropper} $\mathcal{U}_e$. Secure rate $R_{s_{n}}^m$ of user $\mathcal{U}_m$ over subcarrier $n$ is given as \cite{RSaini_CL_2016, Jorswieck2008, Xiaowei_TIFS_2011}: 
	\begin{eqnarray}\label{secure_rate_definition}
	R_{s_{n}}^m = [R_{n}^m- R_{n}^e]^+, 
	\end{eqnarray}
	where $R_{n}^i$ is the rate of the $i$th user over subcarrier $n$. $R_{n}^e$ is the rate of the equivalent eavesdropper $\mathcal{U}_e$ defined as $R_{n}^e = \max\limits_{k\in \{1,2,\cdots, M\} \setminus m} R_{n}^k$~ \cite{Jorswieck2008,Xiaowei_TIFS_2011,RSaini_CL_2016}. $R_{n}^i$ depends on the channel gains of {\color{black}the} $\mathcal{S}-\mathcal{R}$ and $\mathcal{R}-\mathcal{U}$ links, and source and relay powers, $P_{n}^s$ and $P_{n}^r$. Exact rate definition depends on the mode of operation of the relay, e.g. AF or DF (cf. \eqref{sec_rate_af_def_1} and \eqref{sec_rate_df_def_1}).

	Sum secure rate maximization problem with individual power budgets can be stated as:
	\begin{align}\label{opt_prob_rate_max_scp_base}
	& \mathcal{P}0\!: \underset{P_{n,o}^s, P_{n,o}^r, \pi_{n,o}^m} {\text{maximize}} \left[ R_s = \sum_{m=1}^M \sum_{n=1}^N \sum_{o=1}^N  \pi_{n,o}^m  R_{s_{n,o}}^m \right]\!\!,\;\;\text{subject to:}\nonumber \\
	&C_{0,1}: \sum_{m=1}^M\sum_{n=1}^N \pi_{n,o}^m \leq 1 \text{ } \forall o, ~ C_{0,2}: \sum_{m=1}^M\sum_{o=1}^N \pi_{n,o}^m \leq 1 \text{ } \forall n, ~ C_{0,3}: \sum_{n=1}^N  P_{n}^s \leq P_S,\nonumber \\
	& C_{0,4}: \sum_{n=1}^N P_{n}^r \leq P_R,  \qquad\quad C_{0,5}:  \pi_{n,o}^m \in \{0,1\} \text{ } \forall m, n, o, C_{0,6}: P_{n}^s\ge 0, P_{n}^r\ge 0 \text{ } \forall n
	\end{align}
	where $P_S$ and $P_R$, respectively, are the power budgets of $\mathcal{S}$ and  $\mathcal{R}$. $\pi_{n,o}^m$ {\color{black}jointly} indicates subcarrier allocation and subcarrier pairing. $\pi_{n,o}^m=1$, if pair $(n,o)$ is allotted to user $\mathcal{U}_m$; else $\pi_{n,o}^m=0$. $C_{0,1}$ and $C_{0,2}$ ensures that one subcarrier on $\mathcal{S}-\mathcal{R}$ link can be attached with only one subcarrier on $\mathcal{R}-\mathcal{U}$ link. $C_{0,3}$ and $C_{0,4}$ are respective power budget constraints. $C_{0,5}$ shows the binary nature of allocation, and $C_{0,6}$ indicates power variables' non-negativity. 
	
	The joint resource allocation problem {\color{black}involving} subcarrier allocation, subcarrier pairing and power allocation is a combinatorial problem due to the binary {\color{black}nature of} subcarrier allocation and subcarrier pairing variables. With $P_{n,o}^r$ and $P_{n,o}^s$ as real variables, and $\pi_{n,o}^m$ as binary variables, there are $N(M+2)$ variables per subcarrier in the mixed integer non-linear programming (MINLP) problem $\mathcal{P}0$. {\color{black}Optimization problem $\mathcal{P}0$ having exponential complexity with number of subcarriers, belongs to the class of NP hard problems \cite{Y_F_LIU_TSP_2014}. Though globally optimal solution of this problem is the best upper bound on system performance, finding globally optimal solution may not be feasible in polynomial time.}
	
	In order to handle the combinatorial non-convex joint optimization problem $\mathcal{P}0$, we decouple it into three parts by using an equivalent transformation (for global optimal subcarrier allocation), a tight approximation (for near-global-optimal SCP), and the generalized convexity principles (for global-optimal power allocation). 
	In particular, we first solve one of the combinatorial aspect by presenting the global optimal subcarrier allocation policy. Then, the {\color{black}other} combinatorial aspect is dealt by obtaining the optimal SCP policy using a tight approximation. Lastly, the global-optimal power allocation for a given subcarrier allocation, and SCP is obtained by utilizing the generalized convexity principles \cite{Baz}. Next, we present the key insights and reasoning behind this proposed joint resource allocation {\color{black}strategy}.
	
	\subsection{Solution Methodology}
	Since there is no direct ($\mathcal{S}-\mathcal{U}$) link availability, the subcarrier allocation has to be completed over the $\mathcal{R}-\mathcal{U}$ link only. The subcarrier pairing policy matches subcarriers over $\mathcal{S}-\mathcal{R}$ and $\mathcal{R}-\mathcal{U}$ links, which are independent. Thus, in our adopted system model with untrusted users, subcarrier allocation and subcarrier paring are \textit{independent}, which can be \textit{equivalently} decoupled as separate problems. Further, since the same power is received by the user as well as the eavesdropper over a subcarrier, subcarrier allocation which ensures positive secure rate, is also independent of the power allocation policy \cite{RSaini_CL_2016}. \textit{Hence, we summarize that the optimal  subcarrier allocation is an independent decision, which can be done without compromising the joint optimality of the solution to the problem $\mathcal{P}0$.}
	
	Next, we highlight that the other two allocations, {\color{black}namely}, subcarrier pairing, and power allocation are \textit{not independent}, as discussed later in Section \ref{sec_opt_scp}. So, to present an efficient solution methodology, we first assume that SCP is known, and we wish to jointly optimize subcarrier allocation and power allocation variables. In this context, in Section \ref{sec_sca_pa_af} we prove that for a known SCP, source and relay power allocation along with subcarrier allocation can be jointly solved \textit{global-optimally}. 
	Then later in Section \ref{sec_opt_scp}, we present a \textit{tight asymptotic approximation} for optimal SCP utilizing its independence from underlying power allocation in the high SNR regime. 
	This near-optimal SCP is utilized in the joint subcarrier and power allocation problem $\mathcal{P}1$ defined in Section \ref{sec_opt_scp}.  Even after this relaxation in the SCP policy which is based on a tight approximation for effective channel gains to make SCP and power allocation independent of each other, later via Fig. \ref{fig:insight_c_df_af_2u3c} in Section \ref{proposed_optimal_vs_bruteforce} we numerically validate  that the proposed joint solution matches closely with the global-optimal one as obtained after applying brute force over all possible SCP combinations. 
	
	
	\section{Joint Subcarrier and Power Allocation}\label{sec_sca_pa_af} 
	Observing that subcarrier pairing and power allocation are dependent, where as subcarrier allocation is independent of both of them, we consider joint subcarrier allocation and power allocation for a known subcarrier pairing. 
	Thus, we consider default subcarrier pairing where subcarrier $n$ on $\mathcal{S}-\mathcal{R}$ link is paired with subcarrier $n$ on $\mathcal{R}-\mathcal{U}$ link. 
	The secure sum rate maximization problem with this subcarrier pairing is given by
	\begin{align}\label{opt_prob_rate_max}
	& \mathcal{P}1: \underset{P_{n}^s, P_{n}^r, \pi_{n}^m,\,\forall n,m} {\text{maximize}} \left[ R_s = \sum_{m=1}^M \sum_{n=1}^N   \pi_{n}^m  R_{s_{n}}^m \right],\qquad\text{subject to:}  \nonumber \\ 
	&C_{1,1}: \sum_{m=1}^M \pi_{n}^m \leq 1 \text{ } \forall n, ~ C_{1,2}: \pi_{n}^m \in \{0,1\} \text{ } \forall m, n, 
	~C_{1,3}: \sum_{n=1}^N P_{n}^s \leq P_S, \nonumber \\& C_{1,4}: \sum_{n=1}^N P_{n}^r \leq P_R, 
	\quad ~C_{1,5}: P_{n}^s\ge 0, P_{n}^r\ge 0 \text{ } \forall n,
	\end{align}
	where $\pi_{n}^m$ is a subcarrier allocation variable. $C_{1,1}$ and $C_{1,2}$ are  subcarrier allocation constraints. $C_{1,3}$ and $C_{1,4}$ are individual power budget constraints, and $C_{1,5}$ are power variables' non-negativity constraints. $\mathcal{P}1$ is a MINLP having $M+2$ variables per subcarrier, i.e., $\pi_{n}^m$ as $M$ binary variables, and  $P_{n}^r$ and $P_{n}^s$ as real variables. First, we find the optimal subcarrier allocation policy for an AF relayed system after investigating the conditions for achieving positive secure rate over a subcarrier in Section \ref{sec_secure_rate}. Next, by observing the nature of secure rate with respect to source and relay powers in Section \ref{nature_secure_rate}, generalized convexity of the joint power allocation for a given optimal subcarrier allocation and SCP is proved, and the global-optimal power allocation is thus, obtained by solving the KKT conditions in Section \ref{conexity_and_solution_af}. For the sake of completeness, this section closes with a brief summary of the key results on {\color{black}the} optimal subcarrier and power allocations for a DF relayed system \cite{RSaini_CL_2016} in Section \ref{sec_sca_pa_df}. Details of the proposed SCP, and it's independence from the optimal power allocation as obtained using a tight approximation are presented in Section~\ref{sec_opt_scp}.

	\subsection{Optimal Subcarrier Allocation Policy for AF relay}\label{sec_secure_rate}
	In general, subcarrier allocation means allocating $N$ subcarriers among $M$ users. Since there is no direct link, subcarrier allocation has to be done at $\mathcal{R}$ for the subcarriers over the $\mathcal{R}-\mathcal{U}$ link. In {\color{black}an} AF relayed system, $R_{n}^m$ rate of user $\mathcal{U}_m$ over subcarrier $n$ is given as \cite{Jindal_CL_2015}:
	\begin{eqnarray}\label{af_rate_definition}
	R_{n}^m = \frac{1}{2} \log_2 \left \{ 1 + \frac{P_{n}^s\gamma_n^{sr}P_{n}^r\gamma_n^{rm}}{\sigma^2\left(\sigma^2 + P_{n}^s\gamma_n^{sr} + P_{n}^r\gamma_n^{rm}\right)} \right \},
	\end{eqnarray}
	where $\gamma_n^{sr}$ and $\gamma_n^{rm}$ are respective channel gains on $\mathcal{S}-\mathcal{R}$ and $\mathcal{R}-\mathcal{U}$ links over subcarrier $n$,  $\sigma^2$ is additive white Gaussian noise (AWGN) variance. From \eqref{secure_rate_definition}, secure rate positivity condition can be stated as $R_{n}^m > R_{n}^e$. Using \eqref{af_rate_definition} along with monotonicity of $\log(\cdot)$ function:
	\begin{eqnarray}
	\frac{P_{n}^s\gamma_n^{sr}P_{n}^r\gamma_n^{rm}}{\sigma^2 + P_{n}^s\gamma_n^{sr} + P_{n}^r\gamma_n^{rm}} > \frac{P_{n}^s\gamma_n^{sr}P_{n}^r\gamma_n^{re}}{\sigma^2 + P_{n}^s\gamma_n^{sr} + P_{n}^r\gamma_n^{re}}. 
	\end{eqnarray}
	
	After simplifications, feasibility condition for positive secure rate over a subcarrier $n$ is given by $\gamma_n^{rm}>\gamma_n^{re}$. This condition leads to {\color{black}the} optimal subcarrier allocation policy, i.e., allocate a subcarrier to the user having the maximum channel gain. Mathematically,  
	\begin{eqnarray}\label{subcarrier_alloc_relay}
	\pi_{n}^m=
	\begin{cases}
	1 & \text{if $\gamma_{n}^{rm} > \max\limits_{o \in \{1,2,\cdots, M\} \setminus m} \gamma_{n}^{ro}$}\\
	0 & \text{otherwise.}
	\end{cases} 
	\end{eqnarray} 
	
	\begin{remark}
		Condition \eqref{subcarrier_alloc_relay} is necessary and sufficient for positive secure rate over a subcarrier in {\color{black}a} trusted AF relay-assisted OFDMA system with untrusted users. Any other allocation except \eqref{subcarrier_alloc_relay} will lead to  zero secure rate.  It is also worth noting that subcarrier allocation policy for  {\color{black}an} AF relay system is the same as that for {\color{black}a} secure DF relayed system \cite[eq. (5)]{RSaini_CL_2016}.
	\end{remark}
	
	\begin{remark}
		The user having the next best channel gain is the equivalent eavesdropper $\mathcal{U}_e$. Since there is always a user having maximum gain over a subcarrier $n$, each subcarrier is allocated to some user. Thus, \emph{`no communication'} scenario does not arise in untrusted users.
	\end{remark}
	
	To obtain optimal power allocation, first we discuss the unique characteristics of the secure rate, and then prove that the joint power allocation is a generalized convex problem.
	
	\subsection{Nature of Secure Rate in Source and AF Relay Powers}\label{nature_secure_rate}
	After subcarrier allocation using \eqref{subcarrier_alloc_relay}, the secure rate $R_{s_n}^m$ of user $\mathcal{U}_m$ over subcarrier $n$ as defined in \eqref{af_rate_definition}, can be restated as:
	\begin{align}\label{sec_rate_af_def_1}
	&R_{s_n}^m = \frac{1}{2} \log_2 \left \{ 1 + \frac{P_{n}^s\gamma_n^{sr}P_{n}^r\gamma_n^{rm}}{\sigma^2(\sigma^2 + P_{n}^s\gamma_n^{sr} + P_{n}^r\gamma_n^{rm})} \right \} - \frac{1}{2} \log_2 \left \{ 1 + \frac{P_{n}^s\gamma_n^{sr}P_{n}^r\gamma_n^{re}}{\sigma^2(\sigma^2 + P_{n}^s\gamma_n^{sr} + P_{n}^r\gamma_n^{re})} \right \}
	\end{align}
	On applying further simplifications to \eqref{sec_rate_af_def_1}, we obtain:
	\begin{eqnarray}\label{simplified_secure_rate}
	R_{s_n}^m = \frac{1}{2} \log_2 \left[ \frac{(\sigma^2+P_{n}^r\gamma_n^{rm})(\sigma^2 + P_{n}^s\gamma_n^{sr} + P_{n}^r\gamma_n^{re})}{(\sigma^2+P_{n}^r\gamma_n^{re})(\sigma^2 + P_{n}^s\gamma_n^{sr} + P_{n}^r\gamma_n^{rm})} \right]. 
	\end{eqnarray}
	The following proposition outlines the nature of $R_{s_n}^m$ with {\color{black}the} power allocations $P_{n}^s$ and $P_{n}^r$.
	
	\begin{proposition}\label{AF_secure_rate_nature}
		The secure rate $R_{s_n}^m$ of user $\mathcal{U}_m$ over a subcarrier $n$ in {\color{black}an} AF relay-assisted secure communication system is a concave function of $P_{n}^s$, and a pseudoconcave function \cite{Baz} of $P_{n}^r$ achieving a unique maxima. 
	\end{proposition}
	
	\begin{proof}
		Let us denote the operand of the logarithm function in \eqref{simplified_secure_rate} as $\mathcal{O}_n^m$, i.e., 
		\begin{eqnarray}
		\mathcal{O}_n^m = \frac{(\sigma^2+P_{n}^r\gamma_n^{rm})(\sigma^2 + P_{n}^s\gamma_n^{sr} + P_{n}^r\gamma_n^{re})}{(\sigma^2+P_{n}^r\gamma_n^{re})(\sigma^2 + P_{n}^s\gamma_n^{sr} + P_{n}^r\gamma_n^{rm})}.
		\end{eqnarray}
		The first order derivative of $\mathcal{O}_n^m$ {\color{black}with respect to} (w.r.t.) $P_n^s$ is given as:
		\begin{eqnarray}\label{rate_derivative_psn}
		\frac{\partial \mathcal{O}_n^m}{ \partial P_n^s} =  \left(\frac{\sigma^2 + P_{n}^r\gamma_n^{rm}}{\sigma^2 + P_{n}^r\gamma_n^{re}} \right) \frac{P_n^r \gamma_n^{sr} (\gamma_n^{rm}-\gamma_n^{re})}{(\sigma^2 + P_{n}^s\gamma_n^{sr} + P_{n}^r\gamma_n^{rm})^2}.
		\end{eqnarray}
		Since $\gamma_n^{rm}>\gamma_n^{re}$, derivative is always positive. Further, the second order derivative of $\mathcal{O}_n^m$ is always negative (cf. \eqref{eq:Prf1}). Thus, $\mathcal{O}_n^m$ is concave increasing in $P_n^s$. Since $\log(\cdot)$ is a concave increasing function, $R_{s_n}^m$ is concave in $P_n^s$ \cite{Bookboyd}. 
		
		To prove that secure rate is a pseudoconcave function of $P_n^r$, we prove that $R_{s_n}^m$ is unimodal with respect to $P_n^r$. The first order derivative of $\mathcal{O}_n^m$ with respect to  $P_n^r$ is:
		\begin{align}\label{rate_derivative_prn}
		\frac{\partial \mathcal{O}_n^m}{ \partial P_n^r} = \frac{\gamma_n^{sr} P_n^s (\gamma_n^{rm}-\gamma_n^{re}) \left(\sigma^4+\gamma_n^{sr}P_n^s\sigma^2-\gamma_n^{rm}\gamma_n^{re}(P_n^r)^2 \right)}{(\gamma_n^{re}P_n^r + \sigma^2)^2(\gamma_n^{rm} P_n^r + \gamma_n^{sr} P_n^s + \sigma^2)^2}.
		\end{align}
		From \eqref{rate_derivative_prn}, we note that there exists an optimal relay power $P_{n}^{r^*}$ achieving the maximum value of $\mathcal{O}_n^m$. $P_{n}^{r^*} $ obtained by finding critical point of $\mathcal{O}_n^m$ with respect to $P_n^r$, is defined as:
		\begin{eqnarray}\label{optimal_prn_for_psn}
		P_{n}^{r^*} \triangleq \sqrt{\frac{(\sigma^4+P_{n}^s\gamma_n^{sr}\sigma^2)}{\gamma_n^{rm}\gamma_n^{re}}}.
		\end{eqnarray}
		Since $\log(\cdot)$ is a monotonic increasing function, pseudoconcavity of $\mathcal{O}_n^m$ with respect to $P_{n}^{r}$ is retained for $R_{s_n}^m$ also. Thus, secure rate is a pseudoconcave function of $P_{n}^{r}$, with an optimal $P_{n}^{r^*}$ achieving maximum secure rate over $n$.
	\end{proof}
	
	We next obtain the global-optimal power allocation in an AF-relayed secure OFDMA.
	
	\subsection{Generalized Convexity and Power Allocation in AF relay}\label{conexity_and_solution_af}
	After subcarrier allocation, combinatorial aspect of $\mathcal{P}1$ is solved. Joint source and relay power allocation problem for {\color{black}the} AF relay-assisted secure OFDMA system can be stated as:
	\begin{align}\label{opt_prob_rate_max_power_allocation}
	&\mathcal{P}2: \underset{P_{n}^s, P_{n}^r} {\text{max}}\sum_{n=1}^N  \frac{1}{2} \log_2 \left[ \frac{(\sigma^2+P_{n}^r\gamma_n^{rm})(\sigma^2+ P_{n}^s\gamma_n^{sr} + P_{n}^r\gamma_n^{re})}{(\sigma^2+P_{n}^r\gamma_n^{re})(\sigma^2+ P_{n}^s\gamma_n^{sr} + P_{n}^r\gamma_n^{rm})} \right] \nonumber \\ 
	&\text{subject to: } C_{2,1}: \sum_{n=1}^N P_{n}^s \leq P_S, \quad C_{2,2}: \sum_{n=1}^N P_{n}^r \leq P_R,  \quad C_{2,3}: P_{n}^s\ge 0, P_{n}^r\ge 0 \text{ } \forall n.
	\end{align}
	
	The Lagrangian of the problem $\mathcal{P}2$ can be stated as:
	\begin{align}\label{lagrange_opt_prob_rate_max_power_allocation}
	&\mathcal{L}_2 =  \sum_{n=1}^N  \frac{1}{2} \log_2 \left[ \frac{(\sigma^2+P_{n}^r\gamma_n^{rm})(\sigma^2 + P_{n}^s\gamma_n^{sr} + P_{n}^r\gamma_n^{re})}{(\sigma^2+P_{n}^r\gamma_n^{re})(\sigma^2 + P_{n}^s\gamma_n^{sr} + P_{n}^r\gamma_n^{rm})} \right]  -\lambda \left( \sum_{n=1}^N P_{n}^s - P_S \right) - \mu \left( \sum_{n=1}^N P_{n}^r - P_R \right).
	\end{align}
	
	Equating first order derivative of the Lagrangian w.r.t. $P_n^s$ to zero, we get: 
	\begin{eqnarray}\label{lagrange_derivative_with_psn}
	\lambda = \frac{1}{2} \frac{P_n^r \gamma_n^{sr} (\gamma_n^{rm}-\gamma_n^{re})}{(\sigma^2 + P_{n}^s\gamma_n^{sr} + P_{n}^r\gamma_n^{rm})(\sigma^2 + P_{n}^s\gamma_n^{sr} + P_{n}^r\gamma_n^{re})}.
	\end{eqnarray} 
	
	Likewise, equating first order derivative of Lagrangian w.r.t. $P_n^r$ to zero, we get:
	\begin{align}\label{lagrange_derivative_with_prn}
	&\mu = \frac{1}{2} \frac{P_n^s \gamma_n^{sr} (\gamma_n^{rm}-\gamma_n^{re})}{(\sigma^2 + P_{n}^s\gamma_n^{sr} + P_{n}^r\gamma_n^{rm})(\sigma^2 + P_{n}^s\gamma_n^{sr} + P_{n}^r\gamma_n^{re})} \nonumber \\
	& \qquad \qquad . \frac{\left(\sigma^4+\gamma_n^{sr}P_n^s\sigma^2-\gamma_n^{rm}\gamma_n^{re}(P_n^r)^2 \right)}{(\sigma^2 + P_{n}^r\gamma_n^{rm})(\sigma^2 +  P_{n}^r\gamma_n^{re})}
	\end{align} 
	Now 
	we have an important result describing the utilization of source and relay power budgets.

	\begin{proposition}\label{Af_relay_PS_PR_budget}
		In {\color{black}an} AF relay-assisted secure OFDMA system, source power budget is fully utilized. Relay power budget may not be fully utilized, with the allocation such that $P_n^r \leq P_n^{r^*}$.
	\end{proposition}
	
	\begin{proof}
		From \eqref{lagrange_derivative_with_psn} it can be noted that $\lambda=0$ results in $P_n^r = 0$. Since there is no direct connectivity between source and users, $P_n^r = 0$ means `no communication'. Thus, $\lambda>0$, which means that the source power budget constraint will always be active, i.e., $ \sum_{n=1}^N P_{n}^s = P_S $. Hence, the source power budget will be fully utilized.
		
		From \eqref{lagrange_derivative_with_prn} we note that $\mu=0$ results in either $P_n^s = 0$, or $\gamma_n^{rm}\gamma_n^{re}(P_n^r)^2 = \sigma^4+\gamma_n^{sr}P_n^s\sigma^2$. The first condition is a `no communication' scenario, and the second condition indicates optimal relay power allocation (cf. \eqref{optimal_prn_for_psn}) on all subcarriers. Neglecting the possibility of `no communication', $\mu=0$ results in optimal relay power allocation $P_n^r=P_n^{r^*}$ on all subcarriers. $\mu=0$ indicates that the relay power constraint is  inactive, i.e., $\sum_{n=1}^N P_{n}^r < P_R$, which means there is enough relay power budget to allow optimal power allocation on each subcarrier. On the other hand $\mu>0$ indicates $\sum_{n=1}^N P_{n}^r = P_R$, which results in $\gamma_n^{rm}\gamma_n^{re}(P_n^r)^2 < \sigma^4+\gamma_n^{sr}P_n^s\sigma^2$ (cf. \eqref{lagrange_derivative_with_prn}), i.e.,   $P_n^r<P_n^{r^*}$. Thus, relay power allocation is either $P_n^r=P_n^{r^*}$ or $P_n^r<P_n^{r^*}$ depending on the relay power budget $P_R$. These conditions also corroborates the concavity and pseudoconcavity of secure rate with $P_n^s$ and $P_n^r$, respectively, given in Proposition \ref{AF_secure_rate_nature}.
	\end{proof}
	
	Theorem \ref{Concavity_Theorem} below describes that power allocation $\mathcal{P}2$ is a generalized convex problem.
	
	\begin{theorem}\label{Concavity_Theorem}
		In {\color{black}an} AF relayed communication system, joint power allocation is generalized convex in $P_{n}^s$ and $P_{n}^r$ under practical scenarios, and has global optimal solution $P_{n}^{s^*}, P_{n}^{r^*}$.
	\end{theorem}
	\begin{proof}
		See \ref{app1}. 
	\end{proof}
	
	As $\mathcal{P}2$ is a generalized convex optimization problem, its global optimal solution is obtained by solving the KKT conditions. In order to find out {\color{black}the} joint optimal solution $(P_{n}^{s^*}, P_{n}^{r^*})$, first order derivative equations \eqref{lagrange_derivative_with_psn} and \eqref{lagrange_derivative_with_prn} which represent subgradient conditions are to be solved along with the complimentary slackness conditions $\lambda \left( \sum_{n=1}^N P_{n}^s - P_S \right) = 0$, and $\mu \left( \sum_{n=1}^N P_{n}^r - P_R \right) = 0$. Observing \eqref{lagrange_derivative_with_psn} and \eqref{lagrange_derivative_with_prn}, we note that the two conditions are tightly interconnected such that it is not possible to obtain closed form analytical solution for  $P_{n}^s$ and $P_{n}^r$. However, the optimal solution can be found by a two-dimensional search for optimal $\lambda$ and $\mu$, using either subgradient method or any convex problem solver.

	\subsection{Subcarrier and Power Allocation in DF Relay}\label{sec_sca_pa_df}
	For a DF relayed secure cooperative communication system, the resource allocation problem involving subcarrier and power allocation is presented in \cite{RSaini_CL_2016}. The concepts and key contributions have been described in the following subsection. The secure rate in {\color{black}a} DF relayed system over a subcarrier $n$ is 
	\begin{equation}\label{sec_rate_df_def_1}
	R_{s_{n}}^m = \frac{1}{2} \left\{ \min{ (R_{n}^{sr}, R_{n}^{rm} ) } - R_{n}^{re} \right\}^+.
	\end{equation}
	where $R_{n}^{sr}$ and $R_{n}^{rm}$ respectively denote the rates of $\mathcal{S}-\mathcal{R}$ and $\mathcal{R}-\mathcal{U}_m$ link.
	In a DF relayed system, rates $R_{n}^{sr}$ and $R_{n}^{rm}$ are given by 
	$\log_2 \left( 1+ \frac{P_{n}^s\gamma_{n}^{sr}}{\sigma^2}\right)$ and
	$\log_2 \left( 1+ \frac{P_{n}^r\gamma_{n}^{rm}}{\sigma^2}\right)$, 
	respectively. 
	Optimal subcarrier allocation for DF relayed system \cite[eq. (5)]{RSaini_CL_2016} is the same as that of {\color{black}an} AF relayed system \eqref{subcarrier_alloc_relay}.
	The optimization problem for sum secure rate maximization with individual power budgets constraints is given as \cite{RSaini_CL_2016}:
	
	\noindent \begin{align}\label{opt_prob_rate_max_simplified_obj}
	& \mathcal{P}3: \underset{P_{n}^s, P_{n}^r, t_n,\,\forall n} {\text{maximize}} \left[ \widehat{R_{s}} \triangleq  \sum_{n=1}^N \frac{1}{2} \left \{ \log_2 \left( \frac{1+ t_n} {1+ \frac{P_{n}^r \gamma_{n}^{re}}{\sigma^2}} \right) \right \}  \right], ~~\text{subject to: } \nonumber \\ 
	& C_{3,1}: t_n \leq \frac{P_{n}^s\gamma_{n}^{sr}}{\sigma^2} \text{ } \forall n,  \quad C_{3,2}: t_n \leq \frac{P_{n}^r\gamma_{n}^{rm}}{\sigma^2} \text{ } \forall n, \qquad C_{3,3}: \sum_{n=1}^N P_{n}^s \leq P_S,\nonumber \\ 
	& C_{3,4}: \sum_{n=1}^N P_{n}^r \leq P_R, \qquad C_{3,5}: P_{n}^r \gamma_{n}^{re} \le P_{n}^s \gamma_{n}^{sr} \text{ } \forall n 
	\quad ~C_{3,6}: P_{n}^s\ge 0, P_{n}^r\ge 0 \text{ } \forall n  
	\end{align} 
	$C_{3,1}$ and $C_{3,2}$ are from the definition of $min$ function. $C_{3,3}$ and $C_{3,4}$ are power budget constraints. $C_{3,5}$ comes from secure rate positivity constraints, and $C_{3,6}$ captures power budget constraints. Theorem 1 in \cite{RSaini_CL_2016} states that, for energy-efficient optimal power allocation over each subcarrier $t_n = \frac{P_{n}^s \gamma_{n}^{sr}}{\sigma^2} = \frac{P_{n}^r \gamma_{n}^{rm}}{\sigma^2}$. Now,  the power allocation problem gets simplified as: 
	\begin{align}\label{opt_prob_rate_max_simplified_obj2}
	& \mathcal{P}4: \underset{P_{n}^r,\,\forall n} {\text{maximize}} \left[ \sum_{n=1}^N \frac{1}{2} \left \{ \log_2 \left( \frac{\sigma^2+ P_{n}^r \gamma_{n}^{rm}} {\sigma^2+ P_{n}^r \gamma_{n}^{re}} \right) \right \}  \right] \nonumber \\ 
	&\text{subject to: }   C_{4,1}: \sum_{n=1}^N  \frac{P_{n}^r \gamma_{n}^{rm}}{\gamma_{n}^{sr}} \leq P_S, \quad C_{4,2}: \sum_{n=1}^N P_{n}^r \leq P_R, \quad C_{4,3}: P_{n}^r\ge 0 \text{ } \forall n.  
	\end{align} 
	As noted in \cite{RSaini_CL_2016}, $\mathcal{P}4$ belongs to the class of generalized convex problems. KKT conditions can be used to find the optimal solution of $\mathcal{P}4$. The Lagrangian of the problem is:
	\begin{align}\label{lagrange_rate_maximization}
	&\mathcal{L}_4 = \sum_{n=1}^N \frac{1}{2} \left \{ \log_2 \left( \frac{\sigma^2+ P_{n}^r \gamma_{n}^{rm}} {\sigma^2+ P_{n}^r \gamma_{n}^{re}} \right) \right \} - \lambda \left( \sum_{n=1}^N  \frac{P_{n}^r \gamma_{n}^{rm}}{\gamma_{n}^{sr}} - P_S \right) - \mu \left( \sum_{n=1}^N P_{n}^r - P_R \right).
	\end{align}
	Setting the first {\color{black}order} derivative of $\mathcal{L}_4$ w.r.t.  $P_{n}^r$ to zero, we obtain:
	\begin{eqnarray}\label{power_allocation_df_analytical}
	\frac{\sigma^2 \left(\gamma_{n}^{rm}-\gamma_{n}^{re}\right)} {2\left(\sigma^2 + P_{n}^r\gamma_{n}^{rm}\right) \left(\sigma^2 + P_{n}^r\gamma_{n}^{re} \right)} = \mu + \lambda \left( \frac{\gamma_{n}^{rm}}{\gamma_{n}^{sr}} \right) \text{ } \forall n.
	\end{eqnarray}
	
	Optimal $P_{n}^r$ is obtained as a single positive real root of \eqref{power_allocation_df_analytical}, and $P_{n}^s$ is obtained using  $P_{n}^s \gamma_{n}^{sr} = P_{n}^r \gamma_{n}^{rm}$. After obtaining the joint optimal subcarrier allocation and power allocation policies for both AF and DF relayed systems, we next consider the optimal subcarrier pairing policy and its utility in improving the sum secure rate further.
	
	\section{Optimal Subcarrier Pairing as Effective Channel Gain Tailoring}\label{sec_opt_scp}
	The concept of {\color{black}pairing} subcarrier $n$ on $\mathcal{S}-\mathcal{R}$ link with any of the subcarrier $o$ on $\mathcal{R}-\mathcal{U}$ link is referred as subcarrier pairing (SCP). This introduces another degree of freedom, resulting in improved system performance. This performance gain is achieved at the cost of combinatorial aspect added because of SCP, which makes {\color{black}the} joint resource allocation problem more complex. In fact, an $N$ subcarrier based two hop cooperative system has $N!$ possible SCP combinations.
	
	This section presents a near optimal subcarrier pairing scheme designed specifically for secure OFDMA based communication system. An optimal SCP is supposed to match subcarriers on two hops for maximizing the secure rate over each subcarrier. The scheme of pairing sorted gains on $\mathcal{S}-\mathcal{R}$ link with the gains on $\mathcal{R}-\mathcal{U}$ link, proposed for non-secure OFDMA \cite{Hottinen_SPAWC_2007}, and hereafter referred as 'ordered pairing' (OP) is not suitable for secure systems as secure rate definition involves gains of the main user {\color{black}as well as} the eavesdropper. Observing that finding effective channel gain analytically, in a secure OFDMA system, is non-trivial, effective channel gain is obtained in high SNR region. Note that even if water filling schemes for secure OFDMA and normal OFDMA are different, like normal water filling, secure water filling provides more power to the subcarrier with larger effective channel gain. 
	
	SCP can help in improving either {\color{black}the} spectral efficiency, or {\color{black}the} energy efficiency, or both. First, we discuss SCP for a DF relayed system in Section \ref{scp_df} because in this case SCP helps in improving either spectral efficiency or energy  efficiency. We discuss SCP for {\color{black}an} AF relayed communication system in Section \ref{scp_af}, where it improves the spectral efficiency. A brief comparison of resource allocation schemes for {\color{black}both} AF and DF relayed communication systems is presented in Section \ref{comparison_af_df}, and complexity analysis is presented in section \ref{complexity_analysis}.
	
	\begin{remark}
		The term effective channel gain, over a subcarrier pair $(n,o)$, is used to refer to the end to end channel gain  involving the channel gains of both $\mathcal{S}-\mathcal{R}$ and $\mathcal{R}-\mathcal{U}$ links. 
	\end{remark}
	
	\begin{remark} Using conventional definitions of efficiency, improvement in efficiency implies saving of the input resource for achieving a fixed output utility. Thus, improvement in energy efficiency means using lesser power for  realizing a given rate, and improvement in spectral efficiency implies higher secure rate for the same power budget.
	\end{remark}
	
	\subsection{Optimal SCP for DF Relay System}\label{scp_df}
	In order to estimate the benefits of SCP, we need to investigate the possibility of improvement in the system secure rate performance after optimal power allocation. Next, we discuss all power allocation cases in a DF relayed system. From \eqref{power_allocation_df_analytical}, note that depending on relay and source power budgets, there could be three scenarios: (a) $\lambda=0$, $\mu>0$; (b) $\lambda>0$, $\mu=0$; (c) $\lambda>0$, $\mu>0$. In the following we discuss each of these cases in detail.
	\subsubsection{Case (a): $\lambda=0$, $\mu>0$}
	The source power budget constraint is inactive since $\lambda=0$. $\mu>0$ implies that relay power $P_R$ gets used completely and is bottleneck. In this  case, we have following observation. 
	\begin{proposition}\label{prop_df_pr_budget}
In a DF relay assisted secure communication system, the maximum secure rate is controlled by relay power budget $P_R$ when there is enough source power budget $P_S$. In this case, SCP assumes an important role of improving energy efficiency, and helps in achieving secure rate bound by using lesser transmit power. 
	\end{proposition}
	\begin{proof}
		Using $\lambda=0$ in \eqref{power_allocation_df_analytical}, after simplification we get  
		\begin{eqnarray}\label{power_allocation_df_analytical_case_a}
		\frac{\sigma^2(\gamma_{o}^{rm}-\gamma_{o}^{re})} {2\left(\sigma^2 + P_{o}^r\gamma_{o}^{rm}\right) \left(\sigma^2 + P_{o}^r\gamma_{o}^{re} \right)} = \mu,  \text{ } \forall o.
		\end{eqnarray}
		\eqref{power_allocation_df_analytical_case_a} leads to a quadratic in $P_o^r$, and the optimal solution is the positive real root of the quadratic. $\mu$ is obtained such that $\sum_{o=1}^N P_{o}^r = P_R$. Note that, $P_o^r$ depends on  $P_R$ and not on $P_S$. Thus, achievable maximum rate on each subcarrier, and thereby possible maximum sum secure rate depends on $P_R$. Next, we show that this upper bound on secure rate is dependent on available source power and utilization of optimal SCP.
				
		SCP has a limited role in this case because maximum rate is controlled by the $\mathcal{R}-\mathcal{U}$ link. Here SCP can help in achieving the secure rate upper bound but not beyond, thus limited role in spectral efficiency. While SCP has key role in terms of energy efficiency. Through optimal power allocation in DF system, i.e., $P_{n}^s\gamma_{n}^{sr}=P_{o}^r\gamma_{o}^{rm}$ the same SNR is ensured on subcarrier $n$ over $\mathcal{S}-\mathcal{R}$ link and subcarrier $o$ over $\mathcal{R}-\mathcal{U}$ link. Thus, a subcarrier having high $\gamma_{n}^{sr}$ on $\mathcal{S}-\mathcal{R}$ link should be matched with a subcarrier having high $P_{o}^r\gamma_{o}^{rm}$ on $\mathcal{R}-\mathcal{U}$ link, otherwise less  source power will be left for the remaining subcarriers. If SCP is not optimally done then source power budget will get bottlenecked, and the achievable rate will be lower. Thus, SCP helps in better spectral efficiency by obtaining the rate upper bound. 

		Let us discuss the role of SCP in improving energy efficiency through an example. Let there be two subcarriers on $\mathcal{S}-\mathcal{R}$ link with gains such that $\gamma_{1}^{sr}>\gamma_{2}^{sr}$. Over $\mathcal{R}-\mathcal{U}$ link, let the power allocation on these two subcarriers be such that $P_{1}^r\gamma_{1}^{rm}>P_{2}^r\gamma_{2}^{rm}$. Let us consider two pairing scenarios. In first, subcarrier-1 on $\mathcal{S}-\mathcal{R}$ link gets paired with subcarrier-1 on $\mathcal{R}-\mathcal{U}$ link, and in second, subcarrier-1 on $\mathcal{S}-\mathcal{R}$ link gets paired with subcarrier-2 on $\mathcal{R}-\mathcal{U}$ link. Sum source power requirements for the two scenarios are:   		
		
		\begin{eqnarray}\label{power_req_df_case_a_2u2c}
		P_{S_1} = \frac{P_{1}^r\gamma_{1}^{rm}}{\gamma_{1}^{sr}}+\frac{P_{2}^r\gamma_{2}^{rm}}{\gamma_{2}^{sr}}; P_{S_2} = \frac{P_{1}^r\gamma_{1}^{rm}}{\gamma_{2}^{sr}}+\frac{P_{2}^r\gamma_{2}^{rm}}{\gamma_{1}^{sr}}.
		\end{eqnarray} 
		Using \eqref{power_req_df_case_a_2u2c}, $P_{S_{\Delta}}  = P_{S_1} - P_{S_2}$, can be simplified as: 
		\begin{equation}
		P_{S_{\Delta}} = (P_{2}^r\gamma_{2}^{rm}-P_{1}^r\gamma_{1}^{rm})\left( \frac{\gamma_{1}^{sr}-\gamma_{2}^{sr}}{\gamma_{1}^{sr}\gamma_{2}^{sr}} \right) < 0
		\end{equation}
		Note that the source power requirement in second scheme is more. Thus, the  scheme which matches a higher $\gamma_{n}^{sr}$ subcarrier with a subcarrier having higher $P_{o}^r\gamma_{o}^{rm}$ is energy-efficient. 
	\end{proof}
	
	\begin{remark}\label{rm_scp_df_case1}
		For a bottle-necked $P_R$ budget case, optimal SCP matches the sorted $\gamma_{n}^{sr}$ on $\mathcal{S}-\mathcal{R}$ link and  $P_{o}^r\gamma_{o}^{rm}$ on $\mathcal{R}-\mathcal{U}$ link, one by one. This schemes requires lesser source power, and hence energy efficiency gets improved.
	\end{remark}
	
	\subsubsection{Case (b): $\lambda>0$, $\mu=0$}
	$\lambda>0$ implies that source power budget $P_S$ gets fully utilized and is the bottleneck. Relay power budget constraint is inactive. Placing $P_{o}^r$ as $P_{n}^s\gamma_{n}^{sr}/\gamma_{o}^{rm}$ and $\mu=0$, \eqref{power_allocation_df_analytical} is simplified:
	\begin{eqnarray}\label{power_allocation_df_analytical_case_b}
	\frac{\sigma^2(\gamma_{n}^{sr}-\gamma_{n}^{sr'})} {2\left(\sigma^2 + P_{n}^s\gamma_{n}^{sr}\right) \left(\sigma^2 + P_{n}^s\gamma_{n}^{sr'} \right)} = \lambda,  \text{ } \forall n
	\end{eqnarray}
	where $\gamma_{n}^{sr'} \triangleq \frac{\gamma_{n}^{sr}\gamma_{o}^{re}}{\gamma_{o}^{rm}}$. Note the similarity of this equation with  \eqref{power_allocation_df_analytical_case_a}. Optimal $P_{n}^s$ is the positive real root of the quadratic equation obtained from \eqref{power_allocation_df_analytical_case_b}. $\lambda$ is obtained such that $\sum_{n=1}^N P_{n}^s = P_S$. Note that $P_{n}^s$ depends on $\gamma_{n}^{sr'}$ which imbibes SCP.  
	$\gamma_{n}^{sr'}$ depends on which subcarrier $n$ on $\mathcal{S}-\mathcal{R}$ link is paired with which one $o$ on $\mathcal{R}-\mathcal{U}$ link. SCP need to  match $\gamma_{n}^{sr}$ with $\gamma_{o}^{rm}$ and $\gamma_{o}^{re}$ for sum rate maximization.  Thus, SCP maximizes the achievable sum rate, and improves  spectral efficiency. This observation is summarized as follows.
	
	\begin{proposition}\label{prop_df_ps_budget}
		In a DF relayed secure communication system with high $P_R$,  necessary condition to achieve higher secure rate is to match a subcarrier $n$ having higher $\gamma_{n}^{sr}$ on $\mathcal{S}-\mathcal{R}$ link with a subcarrier $o$ having higher $\gamma_{o}^{rm}/\gamma_{o}^{re}$ on $\mathcal{R}-\mathcal{U}$ link. 
	\end{proposition}
	\begin{proof}
		To maximize sum secure rate, power should be allocated in such a way that a subcarrier with higher effective channel gain is allocated higher power. Note that finding effective channel gain $\Gamma_{(n,o)}^{d}$ over a subcarrier pair $(n,o)$ in DF relay system is non-trivial in the general case (cf.  \eqref{power_allocation_df_analytical_case_b}). In the high SNR region, \eqref{power_allocation_df_analytical_case_b} gets simplified as
		\begin{eqnarray}
		\frac{\sigma^2(\gamma_{n}^{sr}-\gamma_{n}^{sr'})} {2(P_{n}^s)^2\gamma_{n}^{sr}\gamma_{n}^{sr'}} = \lambda  \text{ } \forall n.
		\end{eqnarray}
		Thus, the effective channel gain under high SNR scenario reduces to
		\begin{eqnarray}
		\Gamma_{(n,o)}^{d} = \frac{\gamma_{n}^{sr}-\gamma_{n}^{sr'}}{\gamma_{n}^{sr}\gamma_{n}^{sr'}} = \frac{\gamma_{o}^{rm}-\gamma_{o}^{re}}{\gamma_{n}^{sr}\gamma_{o}^{re}} = \frac{1}{\gamma_{n}^{sr}} \left( \frac{\gamma_{o}^{rm}}{\gamma_{o}^{re}}-1 \right). 
		\end{eqnarray}
		
		$\Gamma_{(n,o)}^{d}$ includes channel gains of both the links i.e., $\gamma_n^{sr}$ on $\mathcal{S}-\mathcal{R}$ link, and $\gamma_o^{rm}$ and $\gamma_o^{re}$ on $\mathcal{R}-\mathcal{U}$ link. SCP should be efficiently used to match gains of $\mathcal{S}-\mathcal{R}$ and $\mathcal{R}-\mathcal{U}$ links to find optimal effective channel gains. 
		
		To find optimal SCP in this scenario, let us discuss a simple case of two subcarriers such that $\gamma_{1}^{sr}>\gamma_{2}^{sr}$. To achieve more rate on subcarrier-1, we should have $\Gamma_{1}^{d}>\Gamma_{2}^{d}$ such that $P_{1}^s>P_{2}^s$. With reference to channel gains there exists only  two possibilities either $\frac{\gamma_{1}^{rm}}{\gamma_{1}^{re}} \leq \frac{\gamma_{2}^{rm}}{\gamma_{2}^{re}}$ or $\frac{\gamma_{1}^{rm}}{\gamma_{1}^{re}} > \frac{\gamma_{2}^{rm}}{\gamma_{2}^{re}}$. In the first we have $\frac{\gamma_{1}^{rm}}{\gamma_{1}^{re}}-1 \leq \frac{\gamma_{2}^{rm}}{\gamma_{2}^{re}} -1$, hence $\gamma_{2}^{sr}\left( \frac{\gamma_{1}^{rm}}{\gamma_{1}^{re}}-1 \right) \leq \gamma_{1}^{sr}\left(\frac{\gamma_{2}^{rm}}{\gamma_{2}^{re}} -1\right)$ $\implies$ $\Gamma_{1}^{d} \leq \Gamma_{2}^{d}$. Thus, the only feasible case is $\frac{\gamma_{1}^{rm}}{\gamma_{1}^{re}} > \frac{\gamma_{2}^{rm}}{\gamma_{2}^{re}}$. To ensure, subcarrier with higher effective channel gain is allocated higher power, higher $\gamma_{n}^{sr}$ should be paired with higher $\frac{\gamma_{o}^{rm}}{\gamma_{o}^{re}}$.
	\end{proof}
	
	For $\lambda>0$, $\mu=0$, necessary condition for optimal SCP are presented by Proposition \ref{prop_df_ps_budget}. Next, it is proved that this pairing scheme is  the sufficient condition to improve overall system performance. We present this observation through following theorem. This theorem conceptualizes SCP as \emph{channel gain tailoring} compared to conventional subcarrier mapping. 
	\begin{theorem}\label{theorem_df_scp_channel_tailoring}
		For efficient secure communication, ideal SCP should tailor channel gains such that all subcarriers have same effective channel gain. Practically, optimal SCP wishes to reduce the variance of effective channel gains.
	\end{theorem}
	\begin{proof}
		See \ref{app2}.
	\end{proof}
	The optimal SCP strategy is to tailor the channel gains such that effective channel gains are equal, which leads to equal power allocation, and hence equal rate over all the subcarriers. This may not be feasible as the channels gains are discrete quantities. Hence, a feasible solution is to minimize the variance between the tailored channel gains.

	\begin{remark}
		The optimal SCP in this case, where $P_S$ budget is fully utilized, is to sort $\gamma_{n}^{sr}$ on $\mathcal{S}-\mathcal{R}$ link and $\frac{\gamma_{o}^{rm}}{\gamma_{o}^{re}}$ on $\mathcal{R}-\mathcal{U}$ link, and match one by one. 
	\end{remark}
	\begin{corollary}
		When $P_R$ budget is bottleneck (case (a)), the channel gain tailoring reduces to the SCP strategy of matching $\gamma_{n}^{sr}$ on $\mathcal{S}-\mathcal{R}$ link with $P_{o}^r\gamma_{o}^{rm}$ on $\mathcal{R}-\mathcal{U}$ (cf. remark \ref{rm_scp_df_case1}).
	\end{corollary}
	\subsubsection{Case (c): $\lambda>0$, $\mu>0$}
	$\lambda>0$ and $\mu>0$ is a not a common case, channel gains should be such that $P_n^s\gamma_{n}^{sr} = P_o^r\gamma_{o}^{rm}$ is {\color{black}satisfied} on each subcarrier. Further, both power budgets are tight, i.e., $\sum_{n=1}^N P_{n}^s = P_S$ and $\sum_{o=1}^N P_{o}^r = P_R$. We know that achievable sum secure rate is bounded by $P_R$ (cf. Proposition \ref{prop_df_pr_budget}). Note that optimal SCP matching a higher $\gamma_{n}^{sr}$ with higher $P_o^r\gamma_{o}^{rm}$ is energy-efficient as indicated in the proof of the proposition. The source power budget requirement will be more than the optimal value if optimal SCP is not followed. Then, $P_S$ will {\color{black}become} the effective bottleneck {\color{black}and} the achievable sum secure rate will be lower compared to the upper bound decided by secure water filling on $\mathcal{R}-\mathcal{U}$ link. 
		
	\emph{To summarize, in {\color{black}a} DF relayed system optimal SCP is conditioned on power budgets $P_S$ and $P_R$. Firstly, relay power allocation is done assuming relay power budget $P_R$ is the bottleneck. Subcarriers over $\mathcal{S}-\mathcal{R}$ link sorted according to $\gamma_{n}^{sr}$ are matched with subcarriers over $\mathcal{R}-\mathcal{U}$ link sorted according to $P_o^r\gamma_{o}^{rm}$.  $P_n^s$ is obtained using \cite[Theorem 1]{RSaini_CL_2016}. If $\sum_{n=1}^N P_{n}^s \leq P_S$, then this SCP and power allocation are optimal. Otherwise (if $\sum_{n=1}^N P_{n}^s > P_S$), the actual bottleneck is source power budget $P_S$ and not $P_R$. Now,  subcarriers over $\mathcal{S}-\mathcal{R}$ link sorted according to $\gamma_{n}^{sr}$ are paired with subcarriers over $\mathcal{R}-\mathcal{U}$ link sorted according to $\gamma_{o}^{rm}/\gamma_{o}^{re}$. Source power allocation is achieved by using secure water filling on the $\mathcal{S}-\mathcal{R}$ link.}

	\subsection{Optimal SCP for AF Relay System}\label{scp_af}
	In {\color{black}an} AF relay, secure rate is concave increasing function of source power, and has pseudoconcave nature with relay power (cf. Proposition \ref{AF_secure_rate_nature}). Since optimal relay power $P_o^{r^*}$ is dependent on source power $P_n^s$ (cf. \eqref{optimal_prn_for_psn}), optimal power allocation has to be solved jointly at the source. Power allocation in {\color{black}an} AF relay system is not decomposable as in {\color{black}a} DF relay case, and power allocation cannot be obtained analytically in terms of independent equation of $P_n^s$ and $P_o^r$, due to inter-dependent source and relay power equations (cf. \eqref{lagrange_derivative_with_psn} and \eqref{lagrange_derivative_with_prn}). Thus, finding equivalent channel gain in {\color{black}an} AF relay is more difficult compared to DF case. 
	
	In high SNR regime, \eqref{lagrange_derivative_with_psn} is simplified, and an asymptotic effective channel gain can be estimated. In high SNR scenario, relay uses optimal power over each subcarrier, i.e., $P_o^r = P_o^{r^*}$, such that $\sigma^2(\sigma^2+P_n^s\gamma_n^{sr}) = (P_o^r)^2\gamma_o^{rm}\gamma_o^{re}$. This leads to $\mu=0$. From \eqref{lagrange_derivative_with_psn} we get:  
	\begin{align}\label{effective_channel_gian_AF}
	&\lambda \approx \frac{1}{2} \frac{P_o^r \gamma_n^{sr} (\gamma_o^{rm}-\gamma_o^{re})}{P_{o}^r\gamma_o^{rm} \left( 1 + \frac{P_{o}^r\gamma_o^{re}}{\sigma^2} \right) P_{o}^r\gamma_o^{re} \left(1 + \frac{P_{o}^r\gamma_o^{rm}}{\sigma^2} \right) } \approx \frac{1}{2} \frac{\sigma^4 P_o^r \gamma_n^{sr} (\gamma_o^{rm}-\gamma_o^{re})}{\left( (P_{o}^r)^2\gamma_o^{rm} \gamma_o^{re} \right)^2} \approx \frac{1}{2} \frac{\sigma^2(\gamma_o^{rm}-\gamma_o^{re})} {\sqrt{\sigma^2 (P_{n}^s)^3 \gamma_n^{sr}\gamma_o^{rm} \gamma_o^{re}}}.	
	\end{align} 
	
	From \eqref{effective_channel_gian_AF} we note {\color{black}that}, the effective channel gain over subcarrier pair $(n,o)$ in {\color{black}an} AF case can be specified as $\Gamma_{(n,o)}^{a} = \frac{\gamma_{o}^{rm}-\gamma_{o}^{re}}{\sqrt{\gamma_n^{sr}\gamma_o^{rm} \gamma_o^{re}}}$. A subcarrier with higher effective channel gain is assigned higher source power, and hence it achieves a higher secure rate. As shown in the proof of Proposition \ref{prop_df_ps_budget}, if $\gamma_1^{sr}>\gamma_2^{sr}$ then, to have $P_1^s>P_2^s$ we need to have $\frac{\gamma_{1}^{rm}-\gamma_{1}^{re}}{\sqrt{\gamma_1^{rm} \gamma_1^{re}}}>\frac{\gamma_{2}^{rm}-\gamma_{2}^{re}}{\sqrt{\gamma_2^{rm} \gamma_2^{re}}}$. Thus, a higher $\gamma_n^{sr}$ should be matched with higher $\frac{\gamma_{o}^{rm}-\gamma_{o}^{re}}{\sqrt{\gamma_o^{rm} \gamma_o^{re}}}$ to maximize sum secure rate.
	
	\begin{remark}
		Note that the observation in Theorem \ref{theorem_df_scp_channel_tailoring} is valid for any power allocation strategy that assigns more power over a subcarrier with more channel gain. Since AF power allocation also provides more $P_n^s$ to a subcarrier with higher effective channel gain $\Gamma_{(n,o)}^{a}$, the optimal SCP tries to minimize the variance of the tailored effective channel gains. 
	\end{remark}
	
	\emph{In a nutshell, the asymptotically optimal SCP policy for {\color{black}an} AF relay is to match the sorted $\gamma_n^{sr}$ on $\mathcal{S}-\mathcal{R}$ link with the sorted $\frac{\gamma_{o}^{rm}-\gamma_{o}^{re}}{\sqrt{\gamma_o^{rm} \gamma_o^{re}}}$ ratios on the $\mathcal{R}-\mathcal{U}$ link.} Further, even though the optimal SCP has been investigated under high SNR assumption, its validity at low SNR has been proved through numerical results in section \ref{proposed_optimal_vs_bruteforce}. 
	
	\subsection{Comparison between AF and DF Schemes}\label{comparison_af_df}
	Here we conduct a brief comparison study between the AF and DF relayed systems.
	\subsubsection{Subcarrier Allocation} 
	For both AF and DF relayed system, subcarrier allocation strategy is to  allocate a subcarrier to a user having maximum gain over the $\mathcal{R}-\mathcal{U}$ link.
	\subsubsection{Power Allocation}
	In DF relay, secure rate is increasing in source as well as relay power. Thus, optimum value is achieved at the boundary condition decided by either source or relay power budget. If there is enough $P_S$, and $P_R$ is constrained, the power allocation is solved at the relay through secure water filling on $\mathcal{R}-\mathcal{U}$ link. If instead $P_S$ is bottleneck, power allocation is obtained by  secure water filling on $\mathcal{S}-\mathcal{R}$ link.
	
	In {\color{black}an} AF relayed system, secure rate is concave increasing in source power, and  it is a pseudoconcave function \cite{Baz} of relay power. Source power budget is always fully utilized, while relay power budget may not be fully used. If $\sum_{n=1}^N P_{n}^{r^*} > P_R$, i.e., relay power budget is bottleneck, then secure rate is lower than the achievable rate with sufficient $P_R$.
	
	\subsubsection{Subcarrier Pairing}
	In {\color{black}a} DF relay-assisted system, depending on the source and relay power budgets, SCP can improve either energy efficiency or spectral efficiency of the system. In contrast, in {\color{black}an} AF relay-assisted system SCP is always helpful in improving spectral efficiency of the system.

	\subsubsection{Sum Secure Rate Upper Bound}
	In {\color{black}a} DF relayed system, if $P_R$ is bottleneck, upper bound on sum rate is controlled by secure water filling on $\mathcal{R}-\mathcal{U}$ link. Instead, if $P_S$ is bottleneck then, the bound is decided by secure water filling on $\mathcal{S}-\mathcal{R}$ link.
	In {\color{black}an} AF relayed system, sum rate upper bound is obtained when there is enough relay power to complete optimal allocation i.e., $\sum_{n=1}^N P_{n}^{r^*} \leq P_R$. Otherwise, the sum rate achieved is lesser than the upper bound.   
	
	\subsection{Algorithm Complexity}\label{complexity_analysis}
	{\color{black}Since the secure rate is concave increasing in source and relay powers in a DF relay case, and concave increasing in source power and pseudo-concave in relay power for an AF relay case, the optimal solution is guaranteed, due to inheret nature of the secure rate definitions \cite{Baz}. Thus, the algorithm achieving optimal solution is bound to converge.} 
	In {\color{black}a} DF relay case subcarrier allocation is a search on $M$ channel gains on $\mathcal{R}-\mathcal{U}$ link, power allocation is a one-dimensional (1D) search in either $\lambda$ or $\mu$, and SCP is matching of sorted channel gains. Through decoupling, we have been able to remove the complexity of subcarrier allocation and SCP. Thus, overall complexity of the resource allocation is dominated by complexity of the power allocation which is a 1D search having complexity $\mathcal{O}(N\log N)$ \cite{W_Yu_TCOM_2006}. 
	In {\color{black}an} AF relay case, after replacing $P_n^s$ from \eqref{lagrange_derivative_with_psn} into  \eqref{lagrange_derivative_with_prn}, we get $N$ equations. These along with the complimentary slackness conditions leads to a system of $(N+2)$ equations in $(N+2)$ unknowns $(P_n^r, \lambda, \mu)$. 
	In a special case when, there is enough relay power budget to allow optimal relay power allocation $(P_n^r=P_n^{r^*})$ over each subcarrier, the power allocation is simply a 1D search for optimal $\lambda$, having the complexity as $\mathcal{O}(N\log N)$. 
	
	\begin{figure*}[!t]
		\begin{minipage}{.5\textwidth}
			\centering
			\epsfig{file=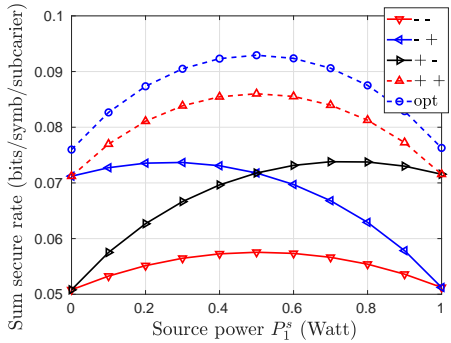,width=2.5in}
			\caption{Maximum sum secure rate, achieved at $P_n^{r^*}$.}
			\label{fig:insight_a_2u2c_af}
		\end{minipage} 
		\begin{minipage}{.5\textwidth}
			\centering
			\epsfig{file=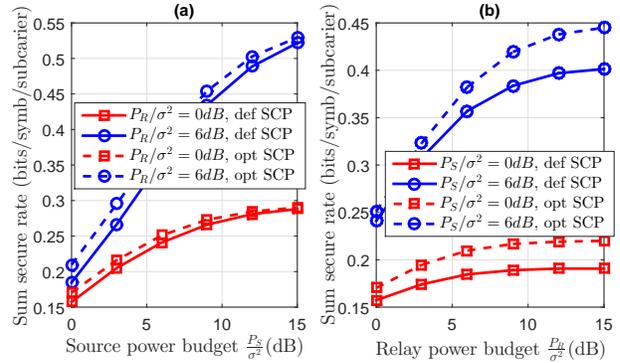,width=3.2in}\vspace{-1mm}
			\caption{Insight on role of SCP in DF relayed system.}
			\label{fig:insiht_b_scp_role_df}
		\end{minipage} 
	\end{figure*} 
		\section{Results and Discussion}\label{sec_results}
The performance of our proposed resource allocation schemes for both AF and DF relayed systems have been presented in this section. By default, downlink of an OFDMA based communication system with $N=64$ subcarriers is assumed. All subcarriers are assumed to bear frequency flat slow fading. The channel parameters remains constant for a frame duration but changes randomly in the next frame. AWGN noise variance is taken as $\sigma^2=0$ dB.  Path loss exponent is assumed to be $3$. $\mathcal{S}$ and $\mathcal{R}$ are, respectively, situated at $(0,0)$ and $(1,0)$. By default $M=8$ untrusted users are considered. Users are randomly placed inside a unit length square which has center at $(2,0)$ {\color{black}as described in Table \ref{sim_param}}. Overall performance is presented in the form of average sum secure rate which is calculated in bits per OFDM symbol per subcarrier (denoted by bits/symb/subcarrier in the figures) averaged over random channel realizations.
	
	First, we present proof of concepts through an exhaustive study over a two-user system in Section \ref{results_insights}. We study DF and AF relayed systems, respectively in Sections \ref{results_df} and \ref{results_af}. 
	Comparison of the proposed schemes with benchmark schemes is discussed in Section \ref{results_comparison}.

\newcolumntype{C}{>{\centering\arraybackslash}p{3em}}
\newcolumntype{L}[1]{>{\raggedright\let\newline\\\arraybackslash}p{#1}}
\begin{table}[!htb]
\caption{System simulation parameters}\label{sim_param}
\centering
{\scriptsize
\begin{tabular}{|L{3cm} |L{1cm}| L{4cm}| L{3cm}| L{1cm}| L{1cm}|}
\hline
System parameter & Symbol & value & System parameter & Symbol & value\\ 
\hline
Number of users & M & 8 & Number of subcarriers & N & 64\\ 
\hline
AWGN variance & $\sigma^2$ & 0 dB & Path loss exponent & $\alpha$ & 3\\ 
\hline
Source location &  &  \{0,0\} & Relay location &  & \{1,0\}\\ 
\hline
Users location &  & Unit square centered at \{2,0\} & & & \\ 
\hline
\end{tabular}
}
\end{table}

	\subsection{Insights in a two-User Secure OFDMA System}\label{results_insights}
	Through this section we provide insights on nontrivial concepts proposed in this work by considering a small, two-user system. These insights help the reader to appreciate usefulness and optimality of the proposed solutions. First we consider relevance of the optimal relay power allocation $P_n^{r^*}$ for AF relayed system, and then the dual role of SCP in a DF system is considered. Next, we show that the proposed SCP based on effective channel gains is as good as finding an exact SCP based on brute force algorithm having exponential complexity. Finally, we present numerical results to corroborate our claim that optimal SCP, denoted as `opt SCP', minimizes the variance of {\color{black}the} effective channel gains.
	
	\subsubsection{Role of $P_n^{r^*}$ in AF Relay Case}
	In order to emphasize the utility of allocating optimal relay power allocation $P_n^{r^*}$, we consider a two-user two-subcarrier system where $P_{1}^s$ is varied from $0$ to $P_S$, and $P_{2}^s = P_S-P_{1}^s$ with $P_S=1$. Corresponding  $P_1^{r^*}$ and $P_2^{r^*}$ are obtained using \eqref{optimal_prn_for_psn}. For a small positive $\delta$, we compare the following relay power allocations in Fig. \ref{fig:insight_a_2u2c_af}:
	\begin{itemize}
		\item Scheme `$--$': with $P_1^{r} = P_1^{r^*} - \delta$, $P_2^{r} = P_2^{r^*} - \delta$
		\item Scheme `$-+$': with $P_1^{r} = P_1^{r^*} - \delta$, $P_2^{r} = P_2^{r^*} + \delta$
		\item Scheme `$+-$': with $P_1^{r} = P_1^{r^*} + \delta$, $P_2^{r} = P_2^{r^*} - \delta$
		\item Scheme `$++$': with $P_1^{r} = P_1^{r^*} + \delta$, $P_2^{r} = P_2^{r^*} + \delta$
		\item Scheme `$opt$': with $P_1^{r} = P_1^{r^*}$, $P_2^{r} = P_2^{r^*}$
	\end{itemize}
	We note that, `$--$' performs worst, as the relay power allocated over both the subcarriers is less than the optimal; `$-+$' and `$+-$' are complimentary schemes having crossover; `$++$' is better than all above, as it has more power on both the subcarriers; while `$opt$' is the best as it allocates optimal relay power. \emph{Thus, allocating optimal power (neither higher nor lower than $P_n^{r^*}$) over each subcarrier is the best strategy for sum rate maximization. }
	
	\subsubsection{Role of SCP in DF Relay Case}
In this subsection, we present the comparison of proposed `opt SCP' with default SCP, denoted by `def SCP', which pairs $n$th subcarrier over $\mathcal{S}-\mathcal{R}$ link with $n$th subcarrier on $\mathcal{R}-\mathcal{U}$ link. We wish to emphasize the roles of SCP in a DF relayed system through a simple two-user two-subcarrier system. Observing Fig. \ref{fig:insiht_b_scp_role_df}, it can be noted that if $P_S$ is high, and $P_R$ is bottleneck, `opt SCP' supports in energy efficiency. Thus, both schemes have same rate. When $P_S$ is bottleneck, and $P_R$ is high, `opt SCP' plays main role in maximizing sum rate. 

	Fig. \ref{fig:insiht_b_scp_role_df}(a) presents sum secure rate performance with $P_S$. For low $P_S$ `opt SCP' plays important role, and improves the sum rate. The gain in sum rate is small at lower $P_R$, whereas it is large at higher $P_R$. For  higher $P_S$, it is $P_R$ which is the bottleneck. It can be noted that, at lower $P_R$ `opt SCP' does not perform well as sum rate curves for `def SCP' and `opt SCP' converge very early compared to higher $P_R$ scenario. Fig. \ref{fig:insiht_b_scp_role_df}(b) captures sum rate performance with relay power $P_R$. For lower $P_S$, `opt SCP' plays a key role in improving sum rate.  \emph{Thus, `opt SCP' depicts important role when $P_S$ is comparatively smaller than $P_R$.}
	
	\begin{figure}[!t]
		\centering
		\epsfig{file=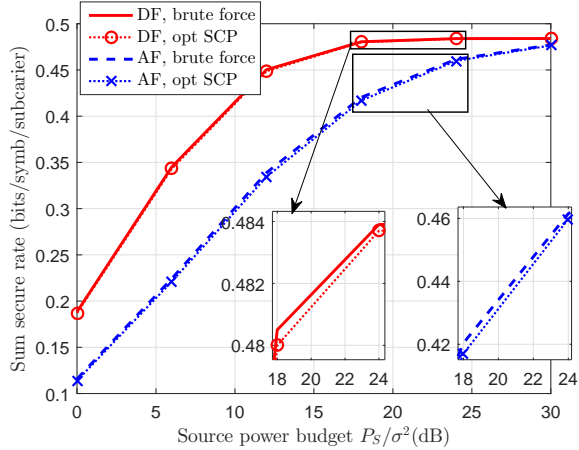,width=3.2in}\vspace{-1mm}
		\caption{Validation of the proposed `opt SCP'.}
		\label{fig:insight_c_df_af_2u3c}\vspace{-2mm}
	\end{figure}
	\subsubsection{Exact versus Asymptotically Optimal SCP}\label{proposed_optimal_vs_bruteforce}
	In Fig. \ref{fig:insight_c_df_af_2u3c}, we compare the best SCP obtained after exhaustive search on all possible combinations with our proposed SCP {\color{black}which is} based on effective channel gains {\color{black}derived in high SNR region}. Instead of a two-user two-subcarrier system which has just two SCP combinations, we consider a two-user three-subcarrier system having $N!=6$ SCP combinations.
	
	The `brute force' scheme chooses the best pairing combination that results in maximum secure rate (among 6 possible combinations) after optimal power allocation. We compare its performance with our proposed `opt SCP', at $P_R/\sigma^2=6$ dB. Note that, secure rate with `opt SCP' is very close to that of `brute force' scheme, and the gap reduces with increasing source power $P_S$. 
	\emph{Thus, `opt SCP', which performs as good as `brute force' (exhaustive  search), is a computationally efficient solution {\color{black}having} reduced complexity by an order of $N!$.}

	\subsubsection{Subcarrier Pairing as Channel Gain Tailoring}
	In this subsection, we validate our claim that optimal SCP tailors channel gains  so as to minimize variance of the effective channel gains. A two-user three-subcarrier system {\color{black}with a DF relay} is considered. There are in total $N!=3!=6$ possible pairing combinations. In Fig. \ref{fig:insight_d_2u3c}(a), effective channel gains of all six subcarriers have been presented. Variance of effective gains and sum rate (at $P_S/\sigma^2 = P_R/\sigma^2 = 15$ dB) are presented in text boxes above each combination. Note that the SCP combination showing minimum gain variance has maximum secure rate. Similar behavior is observed for an AF relay-assisted two-user three-subcarrier systems' performance plotted in Fig. \ref{fig:insight_d_2u3c}(b).  \emph{In summary, the combination having least effective channel gain variance has  the highest secure rate. Thus, `opt SCP' attempts to map subcarriers so as to  achieve effective channel gains having minimum variance.}
		
	\begin{figure}[!t] 
		\centering
		\subfigure[]
		{\includegraphics[width=2.8in]{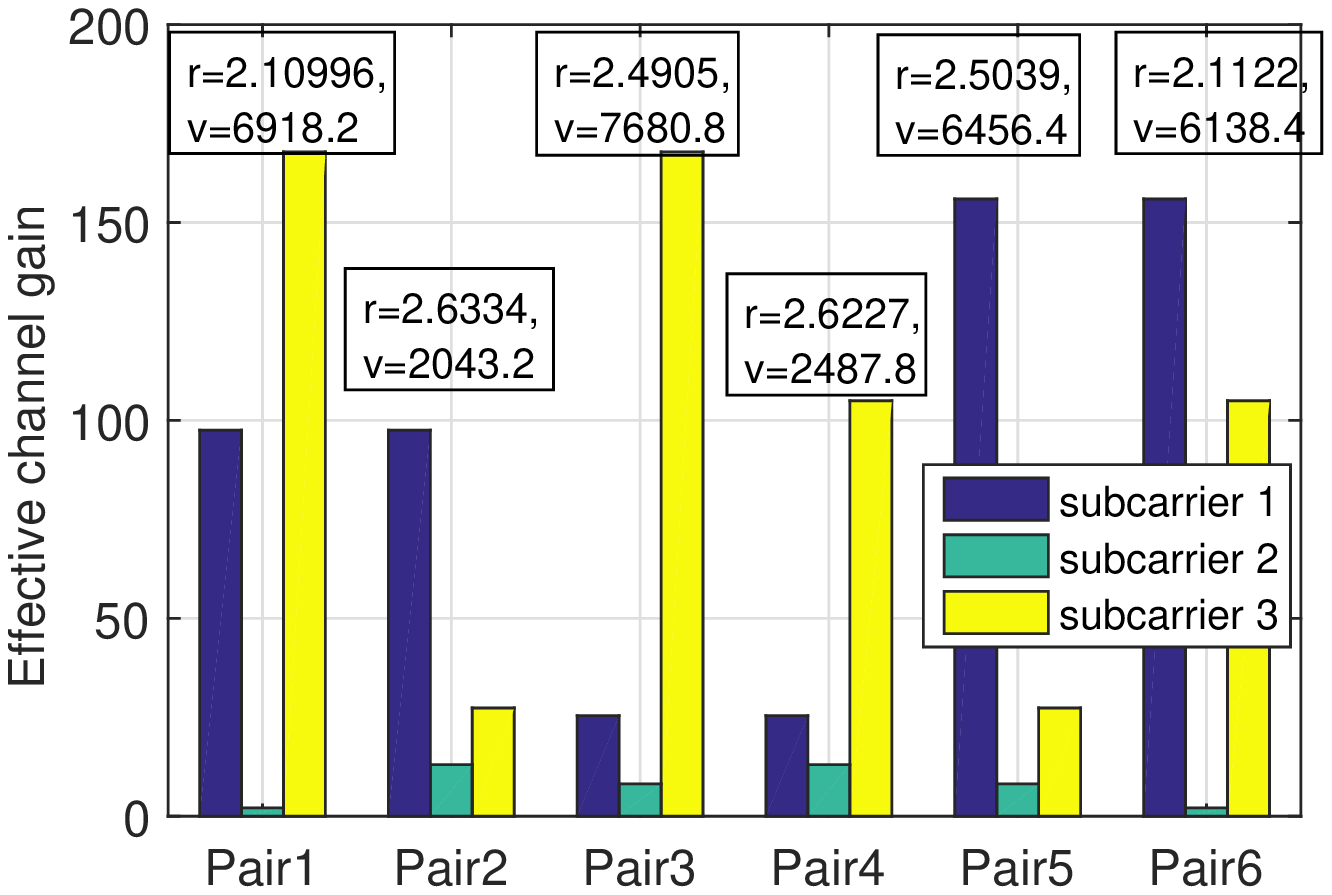} } 
		\subfigure[]
		{\includegraphics[width=2.8in]{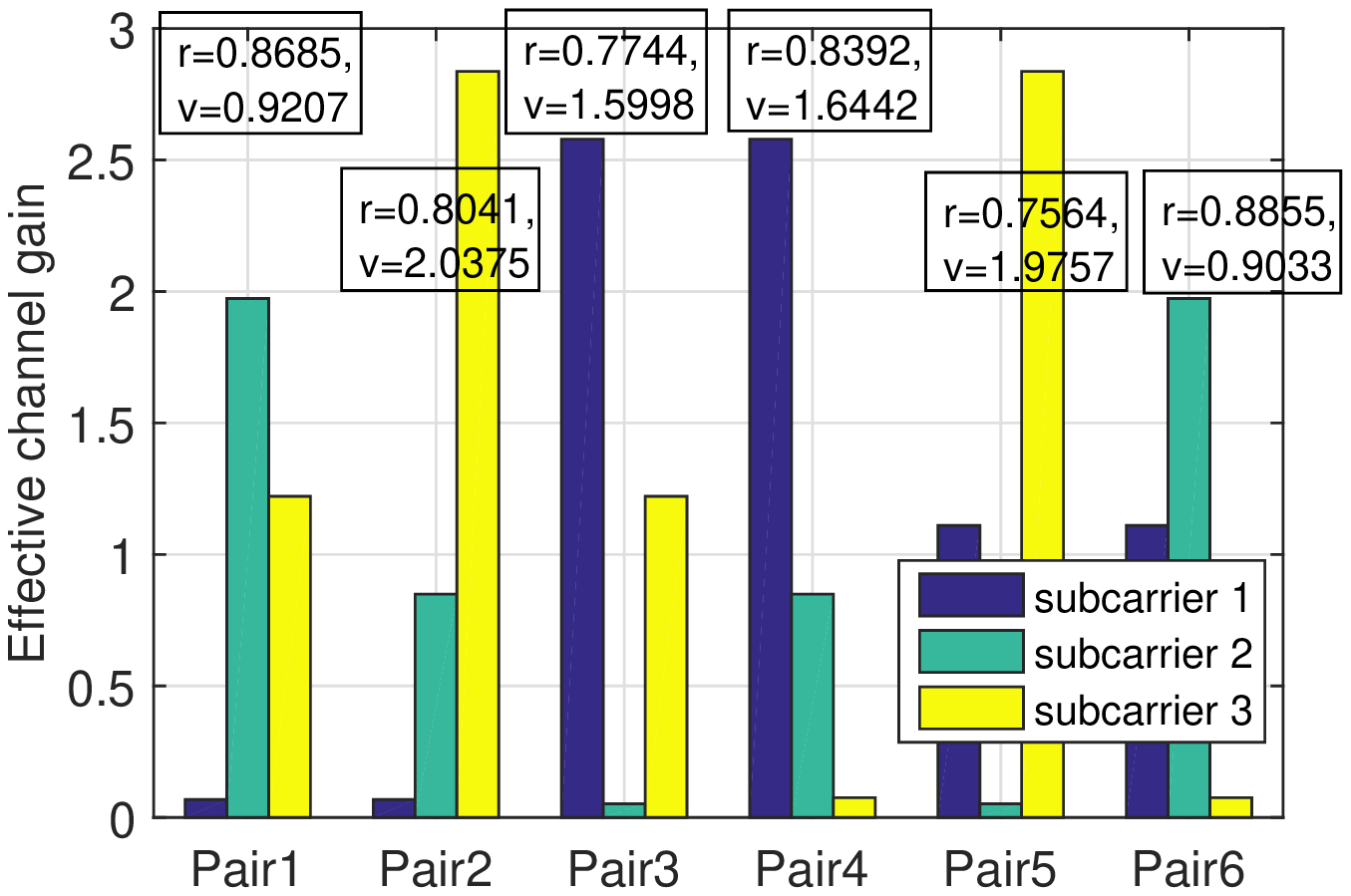} }  
		\caption{Effective channel gains for different SCP  combinations: (a) DF relay system; (b) AF system. `r'
			and `v' respectively denote the sum rate and the variance of effective channel gains for each SCP combination.}
		\label{fig:insight_d_2u3c}
	\end{figure}  
	\subsection{Performance of a DF Relayed System}\label{results_df}
	The performance of a DF relay-assisted secure communications is limited by either $P_S$ or $P_R$. If $P_R$ is the bottleneck then secure rate as provided by `def SCP' cannot be  improved by `opt SCP', whereas if $P_S$ is bottleneck then `opt SCP' plays significant role. In order to highlight the efficacy of `opt SCP' we present sum secure rate with $P_S$ and $P_R$ separately. 
	
	\begin{figure*}[!t]
		\begin{minipage}{.48\textwidth}
			\centering
			\epsfig{file=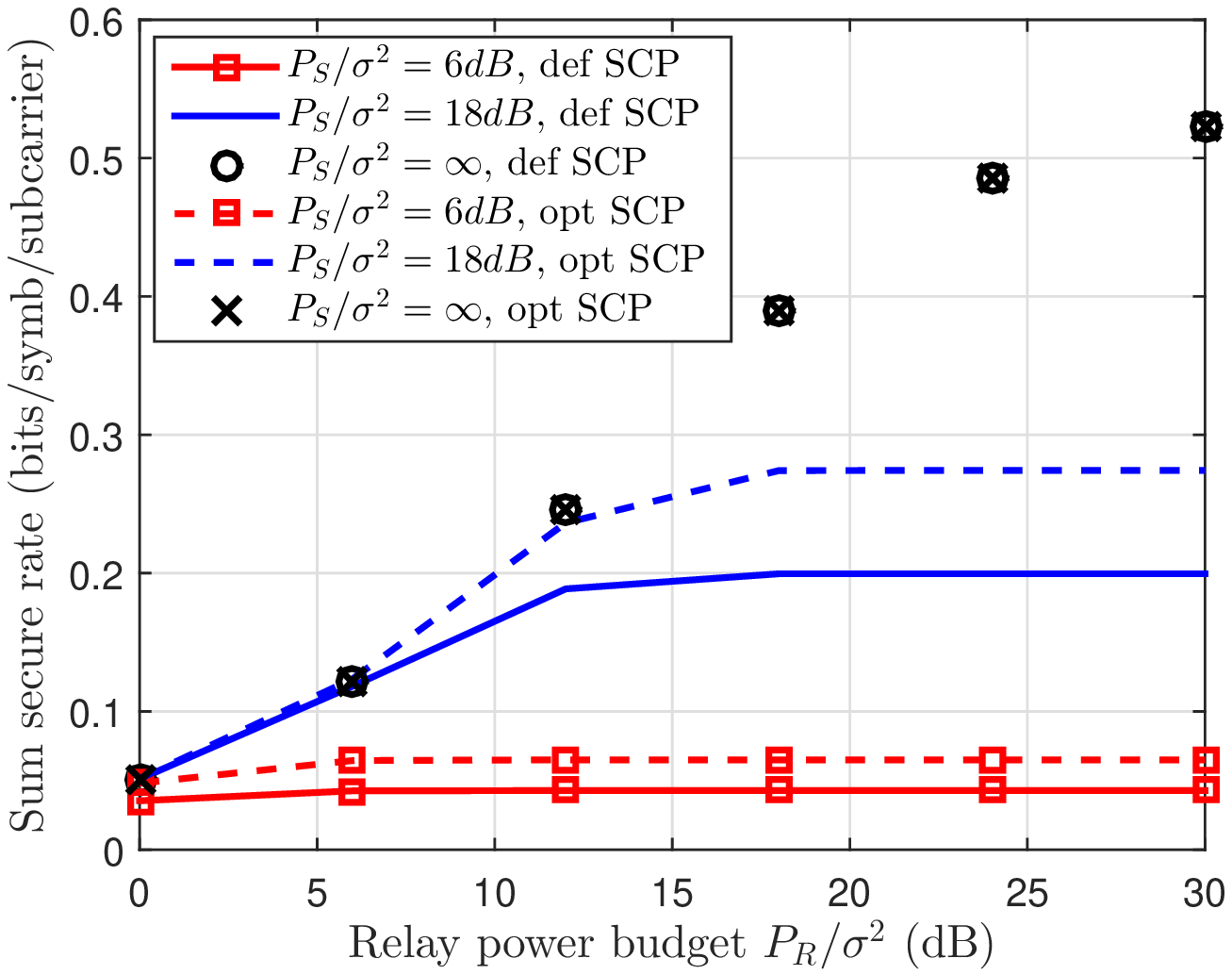,width=3.2in}
			\caption{Variation of sum secure rate versus $P_R$ for different values of $P_S$ in a DF relay-assisted system.}
			\label{fig:sum_rate_max_with_pr_df}
		\end{minipage}\quad 
		\begin{minipage}{.48\textwidth}
			\centering
			\epsfig{file=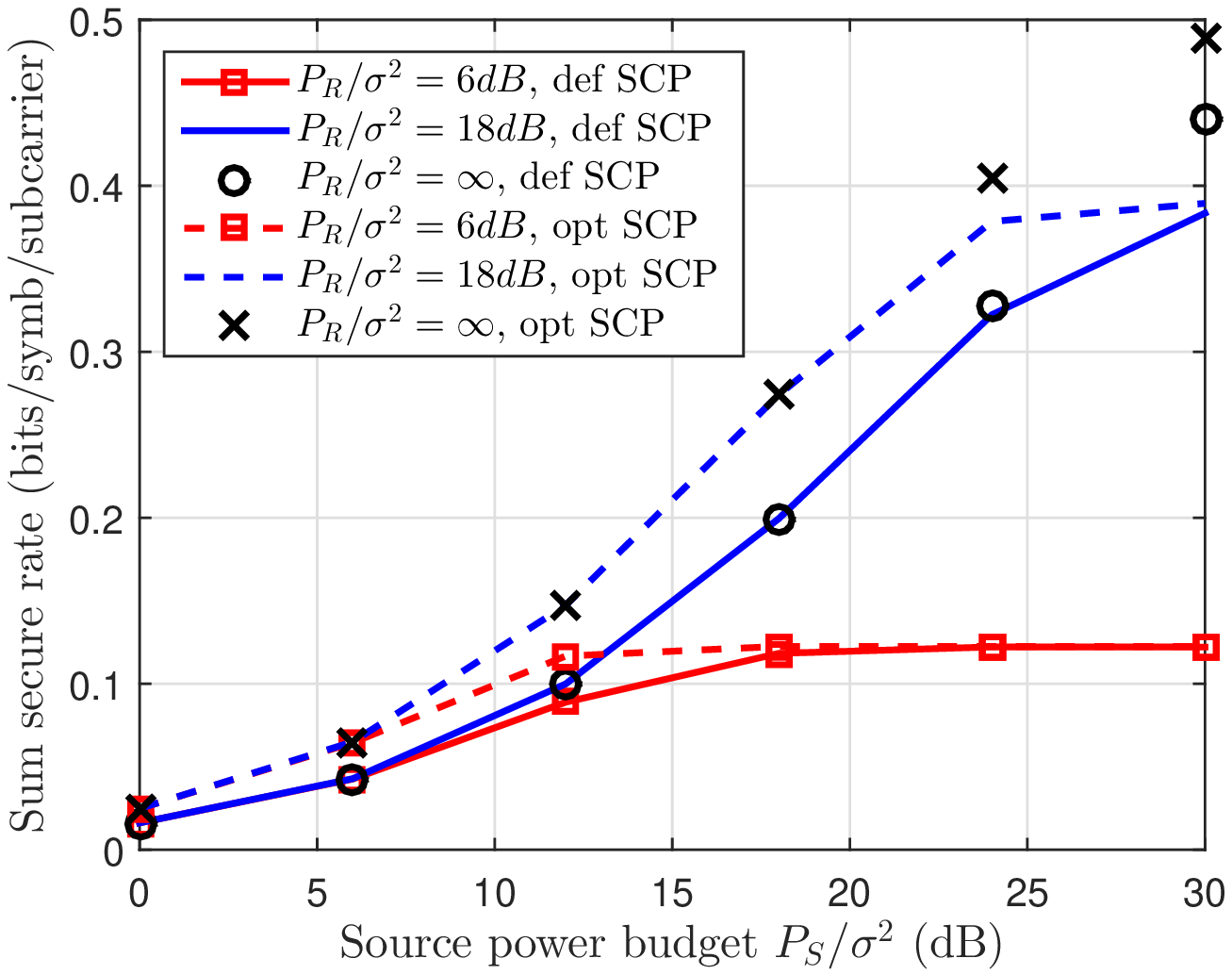,width=3.2in}
			\caption{Variation of sum secure rate versus $P_S$ for different values of $P_R$ in a DF relay-assisted system.}
			\label{fig:sum_rate_max_with_ps_df}
		\end{minipage} 
	\end{figure*}  
	
	\begin{figure*}[!t]
		\begin{minipage}{.48\textwidth}
			\centering
			\epsfig{file=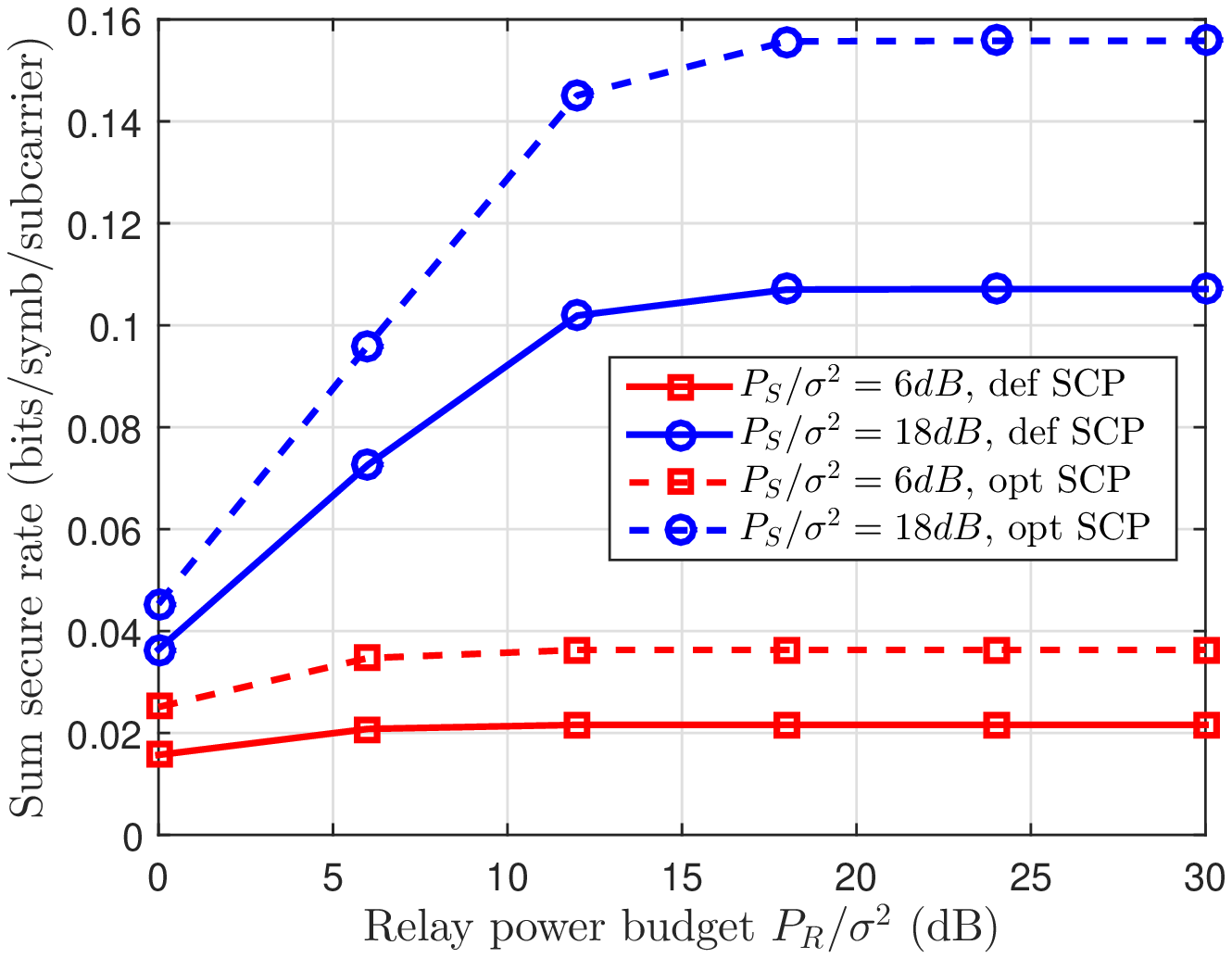,width=3.2in}
			\caption{Variation of sum secure rate versus $P_R$ for different values of $P_S$ in an AF relay-assisted system.}
			\label{fig:sum_rate_max_with_pr_af}
		\end{minipage} \quad
		\begin{minipage}{.48\textwidth}
			\centering
			\epsfig{file=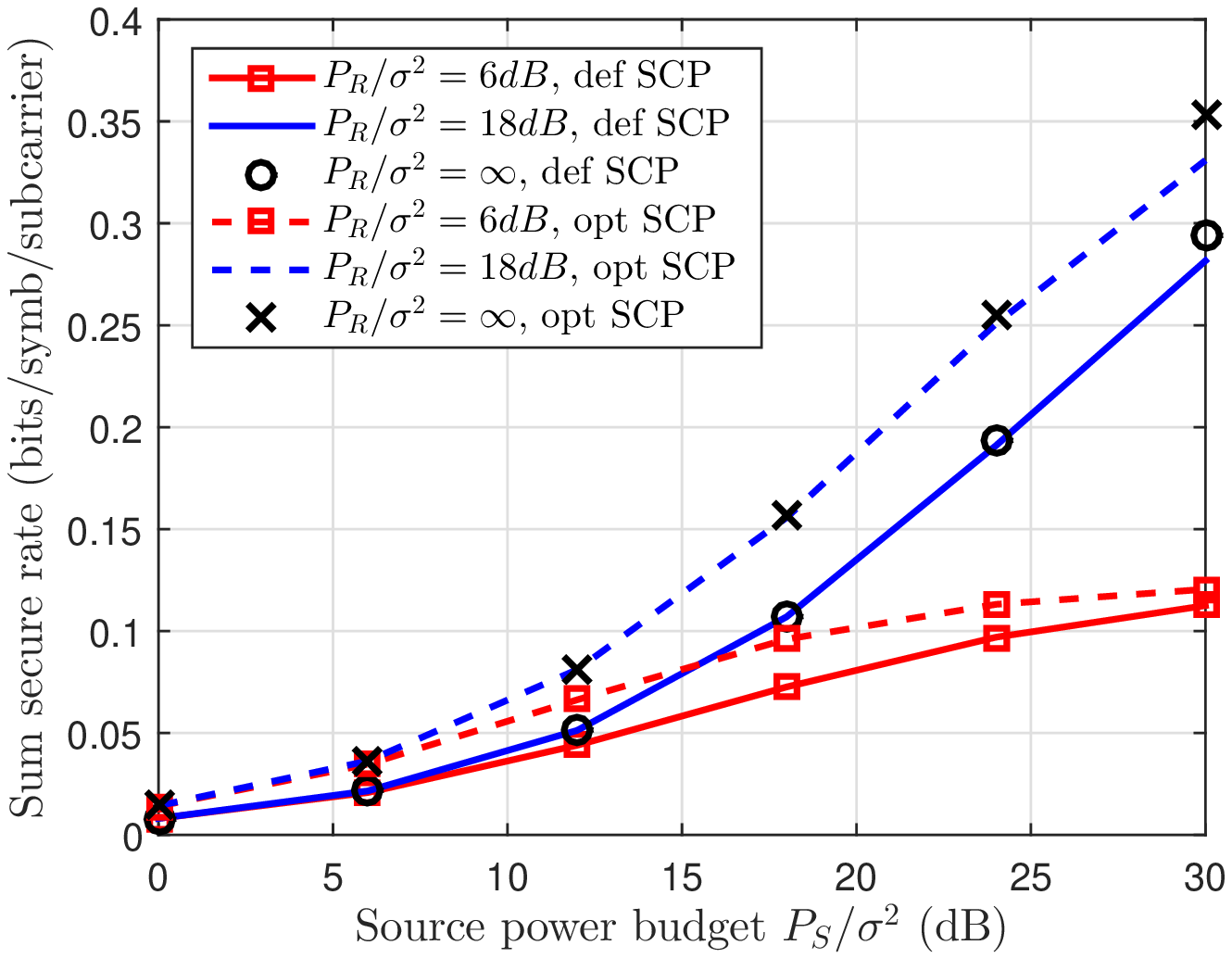,width=3.2in}
			\caption{Variation of sum secure rate versus $P_S$ for different values of $P_R$ in an AF relay-assisted system.}
			\label{fig:sum_rate_max_with_ps_af}
		\end{minipage} 
	\end{figure*}  
	
	\begin{figure}[!t] 
		\centering
		\subfigure[]
		{\includegraphics[width=3in]{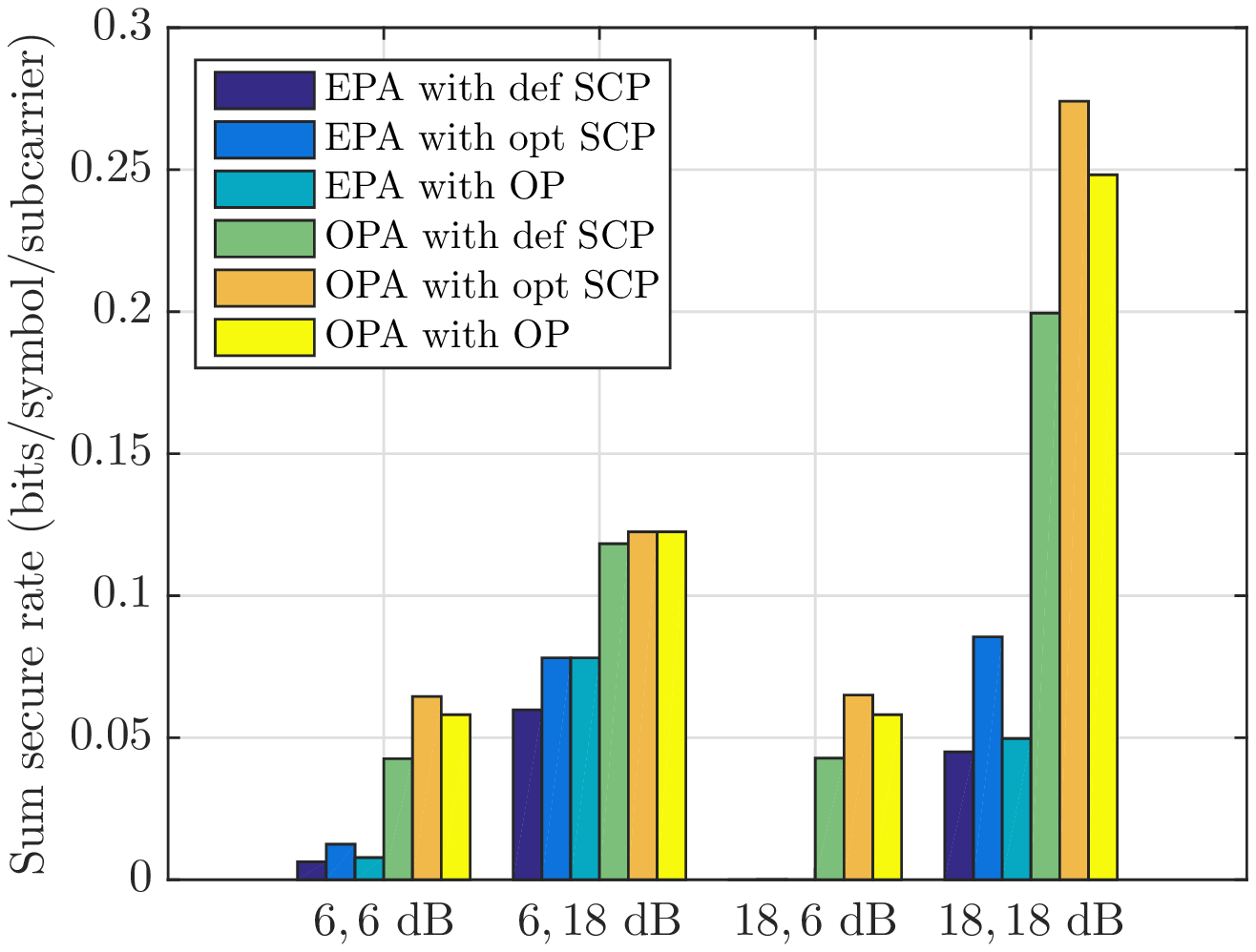} }   
		\subfigure[]
		{\includegraphics[width=3in]{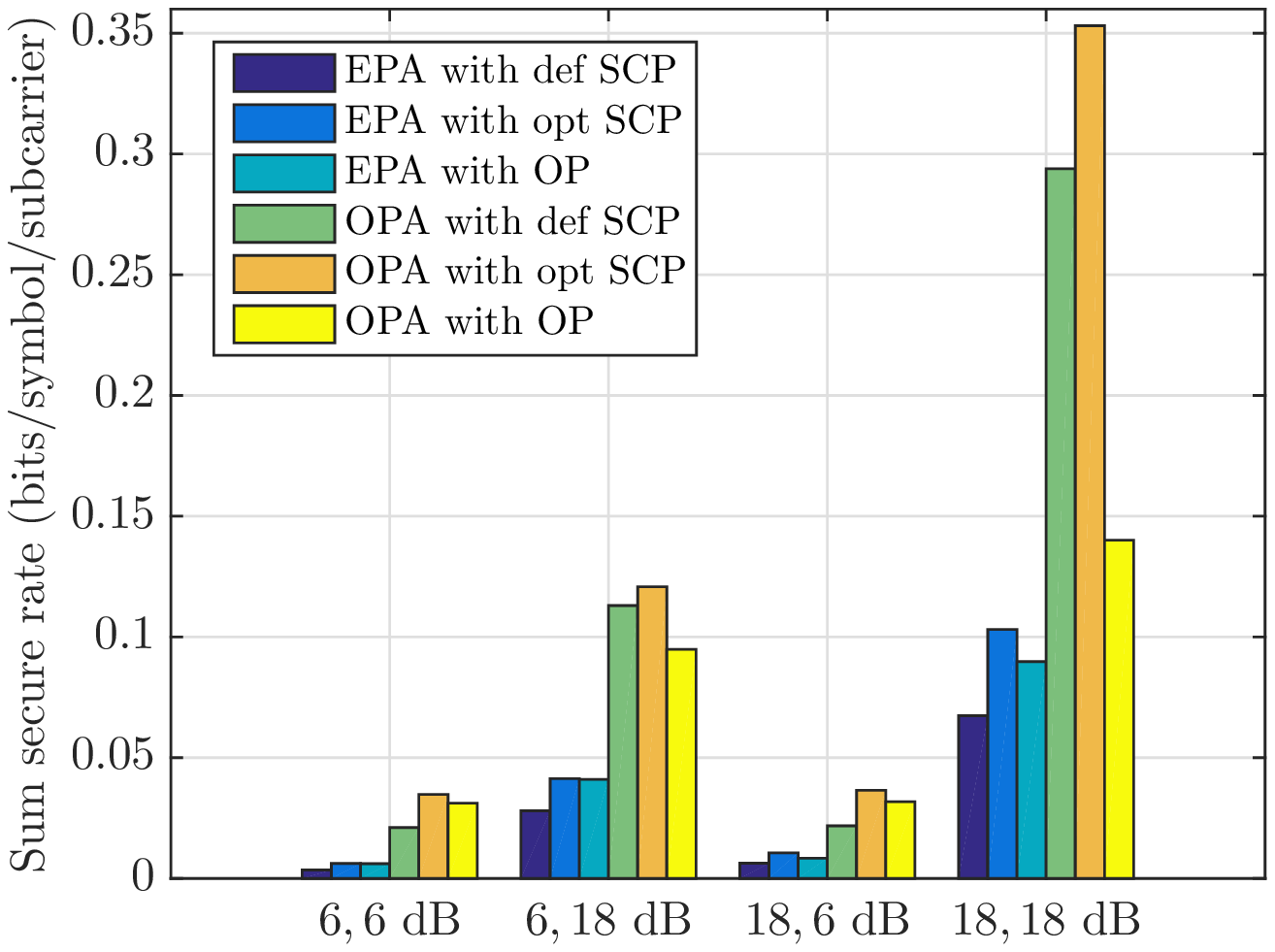} }  \vspace{-2mm}
		\caption{Performance comparison for (a) DF relay and (b) AF relay with different $(P_R,P_S)$.}
		\label{fig:sum_rate_comparison_epa}
	\end{figure}
	In Fig. \ref{fig:sum_rate_max_with_pr_df}, we present the sum rate performance with $P_R/\sigma^2$ for three distinct values of $P_S/\sigma^2$ as $6$ dB, $18$ dB, and $\infty$. $P_S/\sigma^2 = \infty$ is chosen to provide a numerical proof to our proposition that, for a limited $P_R$ `opt SCP' plays a secondary role in improving energy efficiency. When $P_S$ is low, then initially $P_R$ may be the bottleneck but gradually it is $P_S$ which becomes the  bottleneck. Thus at low $P_S$, with varying $P_R$ we observe the change of role by `opt SCP', from energy efficiency to spectrum efficiency, and the associated sum rate improvement. For high $P_S$, this role changes quite late, as initially $P_R$ is bottleneck and `opt SCP' is in complimentary role. It emerges in deciding role when $P_S$ becomes  bottleneck.
	
	Fig. \ref{fig:sum_rate_max_with_ps_df} shows the variation of sum rate with $P_S/\sigma^2$. Three different values of $P_R/\sigma^2$ are taken:  $6$ dB, $18$ dB, and $\infty$.  At a lower $P_R$, with increasing $P_S$, the role of `opt SCP' diminishes from  spectral efficiency to mere energy efficiency, which can be observed from narrowing gap between the performance curves of `def SCP' and `opt SCP'. At high $P_R$, `opt SCP' continues to offer improved overall sum secure rate of the system, and it can be observed by the significant gap between the performance curves of `def SCP' and `opt SCP'.
	
	\subsection{Performance of an AF Relayed System}\label{results_af}
	In this section the performance of {\color{black}an} AF relay-assisted system is considered. The secure rate is concave  increasing with source power while pseudoconcave with relay power. In order, to show the existence of optimal relay power, we show that the secure rate gets saturated with increasing relay power. The performance with relay power budget is plotted in Fig. \ref{fig:sum_rate_max_with_pr_af}. 
	
	The sum rate performance is plotted for two values of source power budget $P_S/\sigma^2=6$ dB and $P_S/\sigma^2=18$ dB. For a lower $P_S$, the optimal relay powers $P_n^{r^*}$ are small. With increasing $P_R$, if there is sufficient $P_R$ to ensure optimal relay power allocation over each subcarrier, the secure rate saturates with $P_R$. At a higher $P_S$, due to higher optimal relay powers $P_n^{r^*}$, the secure rate saturates at a relatively larger $P_R$.
	
	In Fig. \ref{fig:sum_rate_max_with_ps_af}, we present the performance with varying $P_S/\sigma^2$. Here possible values of $P_R/\sigma^2$ are taken as $6$ dB, $18$ dB, and $\infty$ to capture the sum secure rate upper bound that can be achieved by the system. A lower $P_R$ keeps the sum rate bounded as optimal relay power $P_n^{r^*}$ is not available on each subcarrier. With increasing $P_S$, the required  $P_n^{r^*}$ keep increasing, thus the secure rate does not improve much. At higher $P_R$, initially with lower $P_S$ each subcarrier can be allocated optimal relay power so the secure rate improves faster, but later as $P_S$ increases, the secure rate increases slowly as optimal $P_n^{r^*}$ cannot be allocated. Significant performance gain can be observed due to `opt SCP' at higher $P_R$.
	
	\subsection{Comparison with Benchmark Scheme}\label{results_comparison}
	Comparison of presented optimal power allocation (OPA) has been done with equal power allocation (EPA) for `def SCP' and `opt SCP'. To emphasize the benefit of `opt SCP', which is designed for cooperative secure systems, its comparison with `ordered pairing' (OP) \cite{Hottinen_SPAWC_2007} has also been presented. Both EPA and OPA follows optimal subcarrier allocation. This comparison intends to emphasize the performance gain obtained by OPA and `opt SCP'. EPA uses equal power on both  $\mathcal{S}-\mathcal{R}$ and $\mathcal{R}-\mathcal{U}$ links. EPA with `opt SCP' is used to mark the efficacy of `opt SCP' in comparison to `def SCP'. Comparison is done  at four combinations of $(P_R/\sigma^2,P_S/\sigma^2)$ budgets, namely $(6,6)$ dB, $(6,18)$ dB, $(18,6)$ dB, and $(18,18)$ dB. Sum rate performance for DF and AF systems are plotted in Figs. \ref{fig:sum_rate_comparison_epa}(a) and \ref{fig:sum_rate_comparison_epa}(b), respectively.

	Performance comparison for a DF relayed system is presented in Fig. \ref{fig:sum_rate_comparison_epa}(a). Note that for $(P_R/\sigma^2,P_S/\sigma^2) = (6,18)$ dB, `opt SCP' has no significant role, and the gain is minimal. This is because of the limited $P_R$ and in turn bounded achievable rate. `opt SCP' has no scope to improve secure rate further. For this reason OP performs as good as `opt SCP'. For $(P_R/\sigma^2,P_S/\sigma^2) = (18,6)$ dB, EPA gives zero sum rate (for both `def SCP' and `opt SCP'). Since $P_S$ budget is less, it results in $\mathcal{S}-\mathcal{R}$ link to be bottleneck over all subcarriers, forcing all  subcarriers' rate to be zero. Note that joint optimal allocation with OPA and `opt SCP' results in around six times secure rate at $(P_R/\sigma^2,P_S/\sigma^2)=(18,18)$ dB.

	Observing Fig. \ref{fig:sum_rate_comparison_epa}(b) for performance comparison in an AF relayed system, it can be noted that `opt SCP' always plays significant role irrespective of source/relay power budgets. `Opt SCP' offers about 65\% rate improvement compared to optimal power allocation with `def SCP' for $(P_R/\sigma^2,P_S/\sigma^2) = (6,6)$ dB. Even at high SNR such as  $(P_R/\sigma^2,P_S/\sigma^2) = (18,18)$ dB, around 20\% rate improvement is observed in OPA with `opt SCP' over OPA with `def SCP'. OPA with `opt SCP' leads to five times higher secure rate compared to EPA with `def SCP' at $(P_R/\sigma^2,P_S/\sigma^2) = (18,18)$ dB. Further, proposed `opt SCP' shows 10\% and 150\% rate improvement, respectively, for DF and AF relayed system, compared to ordered pairing (OP) for OPA at $(P_R/\sigma^2,P_S/\sigma^2) = (18,18)$ dB. \emph{These results corroborate importance of the proposed `opt SCP' for maximizing sum secure rate in AF and DF relayed systems.}

	\section{Concluding Remarks}\label{sec_conclusion}
	In this paper, joint resource allocation for cooperative secure communication system with untrusted users has been considered. The optimization problem involves subcarrier allocation, power allocation and subcarrier pairing for both AF and DF relayed systems. After presenting optimal subcarrier allocation policy for AF relay,    we have proved that the sum secure rate is concave with source power and pseudoconcave with relay power, and shown that the joint power allocation is a generalized convex problem which can be solved optimally using KKT conditions. SCP has been presented as a novel concept of channel gain tailoring that maximizes sum rate performance. Thus, SCP can be treated as a technique to match subcarriers such that variance between effective channel gains is minimized.
	
	The significant role of optimal SCP in either spectral or energy efficiency improvement has been presented through analytical insights. In a DF relayed system  the sum rate is shown to be controlled by power budget of either source or relay.   Depending on the situation, optimal SCP helps in either improving spectral or energy efficiency. In an AF relay system, optimal SCP plays a crucial role in improving the overall spectral efficiency of the system. We have presented extensive simulation results for a smaller user-subcarrier system to emphasize key concepts. Joint resource allocation scheme is found to give about five and four times improvement  compared to EPA with `deft SCP' for respectively, DF and AF relayed system.


	\appendix
	\section{Proof of Theorem \ref{Concavity_Theorem}} \label{app1}
	We prove the joint-concavity of the secure rate $R_{s_n}^m$ of user $m$ over a subcarrier $n$ in an AF relay-assisted secure communication system in $P_{n}^s$ and $P_{n}^r$ by finding the Hessian matrix, and showing it to be negative semi definite. Let us first define the Hessian matrix $\mathbb{H}\left(\mathcal{O}_n^m\right) \triangleq\left[ \begin{array}{ccc}
	\frac{\partial^2 \mathcal{O}_n^m}{\partial {P_{n}^s}^2} & \frac{\partial^2 \mathcal{O}_n^m}{\partial {P_{n}^s}\partial {P_{n}^r}} \\
	\frac{\partial^2 \mathcal{O}_n^m}{\partial {P_{n}^r}\partial {P_{n}^s}} & \frac{\partial^2 \mathcal{O}_n^m}{\partial {P_{n}^r}^2} \end{array} \right]$ of the operand $\mathcal{O}_n^m$ of $\log(\cdot)$ function in $R_{s_n}^m$, with
	\begin{equation}\label{eq:Prf1}
	\textstyle \frac{\partial^2 \mathcal{O}_n^m}{\partial {P_{n}^s}^2}=-\frac{2 \left(\gamma_n^{sr}\right)^2 P_{n}^r  (\gamma_n^{rm}-\gamma_n^{re}) (\gamma_n^{rm} P_{n}^r +\sigma^2)}{(\gamma_n^{re} P_{n}^r +\sigma^2) (\gamma_n^{rm} P_{n}^r +{\gamma_n^{sr}} P_{n}^s+\sigma^2)^3}.
	\end{equation}
	
	The other second order derivatives, $\frac{\partial^2 \mathcal{O}_n^m}{\partial {P_{n}^s}\partial {P_{n}^r}}$ which is equal to $\frac{\partial^2 \mathcal{O}_n^m}{\partial {P_{n}^r}\partial {P_{n}^s}}$, and $\frac{\partial^2 \mathcal{O}_n^m}{\partial {P_{n}^r}^2}$ are given by \eqref{eq:Prf2} and \eqref{eq:Prf3}. The determinant of $H\left(\mathcal{O}_n^m\right)$ is given by \eqref{eq:Prf4}. 
	From \eqref{eq:Prf4}, it can be observed that 
	$det \left[\mathbb{H}\left(\mathcal{O}_n^m\right)\right] \ge 0$,
	provided   
	$\bigg\{\left(\gamma_n^{rm} \gamma_n^{re} \left(\frac{P_{n}^r}{\sigma^2}\right)^2\ge 1\right)  \wedge \left(\gamma_n^{rm}\ge \gamma_n^{re}\right) \wedge \left(P_{n}^r\le P_{n}^{r^*} \right) \bigg\}$.

	\begin{figure*}		
		\begin{equation}\label{eq:Prf2}\textstyle
		\frac{\partial^2 \mathcal{O}_n^m}{\partial {P_{n}^s}\partial {P_{n}^r}}=\frac{\partial^2 \mathcal{O}_n^m}{\partial {P_{n}^r}\partial {P_{n}^s}}=\frac{{\gamma_n^{sr}}  \left({\gamma_n^{sr}} P_{n}^s \left(\gamma_n^{rm} \gamma_n^{re} \left(P_{n}^r\right)^2+2 \gamma_n^{rm} P_{n}^r\sigma^2 +\sigma^4\right)-(\gamma_n^{rm} P_{n}^r+\sigma^2) \left(\gamma_n^{rm} \gamma_n^{re} \left(P_{n}^r\right)^2-\sigma^4\right)\right)}{(\gamma_n^{rm}-\gamma_n^{re})^{-1}(\gamma_n^{re} P_{n}^r +\sigma^2)^2 (\gamma_n^{rm} P_{n}^r +{\gamma_n^{sr}} P_{n}^s+\sigma^2)^3},
		\end{equation}
		\begin{equation}\label{eq:Prf3}\textstyle
		\frac{\partial^2 \mathcal{O}_n^m}{\partial {P_{n}^r}^2}=-\frac{2 {\gamma_n^{sr}} P_{n}^s  \left(\gamma_n^{rm} \gamma_n^{re} P_{n}^r  \left(3\sigma^2 ({\gamma_n^{sr}} P_{n}^s+\sigma^2)-\gamma_n^{rm} \gamma_n^{re} \left(P_{n}^r\right)^2\right)+\sigma^2({\gamma_n^{sr}} P_{n}^s+\sigma^2) (\gamma_n^{rm}\sigma^2+\gamma_n^{re} ({\gamma_n^{sr}} P_{n}^s+\sigma^2))\right)}{(\gamma_n^{rm}-\gamma_n^{re})^{-1}(\gamma_n^{re} P_{n}^r +\sigma^2)^3 (\gamma_n^{rm} P_{n}^r +{\gamma_n^{sr}} P_{n}^s+\sigma^2)^3}.
		\end{equation}
		\begin{equation}\label{eq:Prf4}\textstyle
		det \left[\mathbb{H}\left(\mathcal{O}_n^m\right)\right] = \frac{\left(\gamma_n^{sr}\right)^2  \left(\left(\gamma_n^{rm} \gamma_n^{re} \left(P_{n}^r\right)^2-\sigma^4\right) \left(-\gamma_n^{rm} \gamma_n^{re} \left(P_{n}^r\right)^2+4 {\gamma_n^{sr}} P_{n}^s\sigma^2+\sigma^4\right)+4 {\gamma_n^{sr}} P_{n}^s \sigma^4(\gamma_n^{re} P_{n}^r +\sigma^2)\right)}{(\gamma_n^{rm}-\gamma_n^{re})^{-2}(\gamma_n^{re} P_{n}^r +\sigma^2)^4 (\gamma_n^{rm} P_{n}^r +{\gamma_n^{sr}} P_{n}^s+\sigma^2)^4}.
		\end{equation} 
	\end{figure*}
	
	
	From \eqref{eq:Prf1} and \eqref{eq:Prf3}, $\frac{\partial^2 \mathcal{O}_n^m}{\partial {{P_{n}^s}^2}}, \frac{\partial^2 \mathcal{O}_n^m}{\partial {P_{n}^r}^2}\le0 \hspace{1mm}\forall\{\left(P_{n}^r\le P_{n}^{r^*}\right)\wedge \left(\gamma_n^{rm}\ge \gamma_n^{re}\right)\}$. This along with the conditions of $det \left[\mathbb{H}\left(\mathcal{O}_n^m\right)\right] \ge 0$ prove that $\mathbb{H}\left(\mathcal{O}_n^m\right)$ is negative semi-definite. Hence, $\mathcal{O}_n^m$ is jointly-concave in $P_{n}^s$ and $P_{n}^r$. Further, as $\log$ is a concave increasing function, $R_{s_n}^m=\frac{1}{2}\log_2\left(\mathcal{O}_n^m\right)$ is also a jointly-concave function of $P_{n}^s$ and $P_{n}^r;$ $\textstyle \forall \bigg\{\left(\gamma_n^{rm} \gamma_n^{re} \left(\frac{P_{n}^r}{\sigma^2}\right)^2\ge 1\right)\wedge \left(P_{n}^r\le P_{n}^{r^*}\right)\wedge \left(\gamma_n^{rm}\ge \gamma_n^{re}\right) \bigg\}$. $\gamma_n^{rm}\ge \gamma_n^{re}$ is ensured by optimal subcarrier allocation policy (cf. \eqref{subcarrier_alloc_relay}). Optimal power allocation ensures $P_{n}^r \leq P_{n}^{r^*}$ (cf. Proposition \ref{Af_relay_PS_PR_budget}). Since under all practical conditions for communication, SNR $\frac{P_{n}^r\gamma_n^{rm}}{\sigma^2}>1$, for any subcarrier $n$, $ \frac{P_{n}^r\sqrt{\gamma_n^{rm} \gamma_n^{re}}}{\sigma^2}$ is practically greater than one.

	Finally, concave objective function represented by a sum of concave functions $R_{s_n}^m\,\forall n$, along with affine constraints $C_{2,1}, C_{2,2}, C_{2,3}$ indicates that the sum rate maximization problem $\mathcal{P}2$ is a generalized convex problem; its global-optimal solution is given by the KKT point.

	\section{Proof of Theorem \ref{theorem_df_scp_channel_tailoring}} \label{app2}
Let us consider a two subcarrier system. The SNRs on subcarriers are $\alpha_1$ and $\alpha_2$ such that $\alpha_1>\alpha_2$. Next we intend to find the effect of widening the gap between the SNRs. Let the updated SNRs be $\alpha_1+x$ and $\alpha_2-x$. Noting the monotonicity of $\log$ function, let us observe the difference $(1+\alpha_1)(1+\alpha_2)$ $-$ $(1+\alpha_1+x)(1+\alpha_2-x)$ which is equal to $(\alpha_1-\alpha_2)x + x^2$. This quantity is always positive for all $x>0$. Thus, there is a positive difference in the sum rate. Hence, widening has resulted in reduced sum rate. It means for maximum sum rate both subcarriers should have equal SNR. 

	Next, we verify the same concept for water filling. Let us take two subcarriers  having gains $\gamma_1$ and $\gamma_2$, such that $\gamma_1>\gamma_2$. The water filling procedure gives power allocation $P_1$ and $P_2$ over the two subcarriers such that
	\begin{eqnarray}
	\textstyle \frac{\gamma_1}{\sigma^2+\gamma_1P_1} = \frac{\gamma_2}{\sigma^2+\gamma_2P_2} = \lambda.
	\end{eqnarray} 
	Substituting $P_2 = P_S-P_1$, we obtain $P_1 = \frac{P_S }{2}+\zeta$ and $\quad P_2 = \frac{P_S }{2}- \zeta$, 
	where $\zeta = \frac{\sigma^2(\gamma_1-\gamma_2)}{2\gamma_1\gamma_2}$. The gap between {\color{black}the} two powers is given as $P_{\Delta} = P_1 - P_2  = \sigma^2 \left( \frac{1}{\gamma_2}-\frac{1}{\gamma_1}\right)$. 
	
	Next, we discuss the other case, where channel gains are changed so as to widen the gap between them. Let the updated channel gains be $\gamma_1' = \gamma_1+\delta$ and $\gamma_2' = \gamma_2-\delta$ for some $\delta>0$. Under these conditions, the updated powers can be written as: $ P_1' = \frac{P_S }{2}+\zeta'$ and $P_2 = \frac{P_S }{2}-\zeta'$, 
	where $\zeta' = \frac{\sigma^2(\gamma_1-\gamma_2+2\delta)}{2(\gamma_1+\delta)(\gamma_2-\delta)}$.
	The difference between the powers is $P_{\Delta}' = P_1' - P_2' = 2\zeta' = \sigma^2 \left( \frac{1}{\gamma_2-\delta}-\frac{1}{\gamma_1+\delta} \right)$. Note that $\frac{1}{\gamma_2-\delta}-\frac{1}{\gamma_1+\delta}>\frac{1}{\gamma_2}-\frac{1}{\gamma_1}$ i.e., $P_{\Delta}'>P_{\Delta}$. Because of the symmetrical powers around $\frac{P_S}{2}$, we can claim that $P_1'>P_1$ and $P_2'<P_2$. Since $P_1'>P_1$ and $(\gamma_1+\delta)>\gamma_1$,  $P_1'(\gamma_1+\delta)>P_1\gamma_1$ and similarly $P_2'(\gamma_2-\delta)<P_2\gamma_2$. Thus, the gap between the respective SNRs has widened because of {\color{black}the} increased gap between channel gains. 
	
	Next we show that by increasing the gap between channel gains the sum rate gets reduced. Let $R_1$ and $R_2$ denote the rates over the subcarriers $1$ and $2$ under the first scheme, where $R_1 = \log_2 \left( 1+\frac{P_1\gamma_1}{\sigma^2} \right) $ and $R_2 = \log_2 \left( 1+\frac{P_2\gamma_2}{\sigma^2} \right)$. The sum rate is
	\begin{align}
	&R = \log_2 \left \{ \left( 1+\left(\frac{P_S }{2}+\frac{\sigma^2(\gamma_1-\gamma_2)}{2\gamma_1\gamma_2}\right)\frac{\gamma_1}{\sigma^2} \right) \right \} \nonumber \\
	& \qquad + \log_2 \left \{ \left( 1+\left(\frac{P_S }{2}-\frac{\sigma^2(\gamma_1-\gamma_2)}{2\gamma_1\gamma_2}\right) \frac{\gamma_2}{\sigma^2} \right) \right\}
	\end{align}
	
	After a few simplifying steps, it can be restated as
	\begin{align}
	R = \log_2 \left[ \frac{\left( \sigma^2 \left(\gamma_1+\gamma_2\right) +P_S\gamma_1\gamma_2\right) ^2}{4\sigma^2\gamma_1\gamma_2} \right]. 
	\end{align}
	
	Let the rates be denoted as $R_1'$ and $R_2'$ in the other scenario with widened channel gains, where $R_1' = \log_2 \left( 1+\frac{P_1'\gamma_1'}{\sigma^2} \right)$ and $R_2' = \log_2 \left( 1+\frac{P_2'\gamma_2'}{\sigma^2} \right)$. The corresponding sum is given as 
	\begin{eqnarray}
	\textstyle  R' = \log_2 \left[ \frac{\left( \sigma^2 \left(\gamma_1+\gamma_2\right) + P_S(\gamma_1+\delta)(\gamma_2-\delta)\right)^2}{4 \sigma^2 (\gamma_1+\delta)(\gamma_2-\delta)} \right].
	\end{eqnarray}
	
	The condition for $R>R'$ can be stated as follows
	\begin{eqnarray}
	\textstyle \frac{\left( \sigma^2 \left(\gamma_1+\gamma_2\right) +P_S\gamma_1\gamma_2\right) ^2}{4\sigma^2\gamma_1\gamma_2} > \frac{\left( \sigma^2 \left(\gamma_1+\gamma_2\right) + P_S(\gamma_1+\delta)(\gamma_2-\delta)\right)^2}{4 \sigma^2 (\gamma_1+\delta)(\gamma_2-\delta)}. 
	\end{eqnarray}
	
	This condition can be further simplified as:
	\begin{align}
	& (\gamma_1+\delta)(\gamma_2-\delta)\left( \sigma^2 \left(\gamma_1+\gamma_2\right) +P_S\gamma_1\gamma_2\right) ^2 > \nonumber \\
	& \qquad \gamma_1\gamma_2 \left( \sigma^2 \left(\gamma_1+\gamma_2\right) + P_S(\gamma_1+\delta)(\gamma_2-\delta)\right)^2.
	\end{align}
	
	After a few more simplification steps, we get
	\begin{align}
	&\sigma^4 (\gamma_1+\gamma_2)^2(\gamma_1+\delta)(\gamma_2-\delta) + (P_S\gamma_1\gamma_2)^2(\gamma_1+\delta)(\gamma_2-\delta) \nonumber \\
	& ~ > \sigma^4(\gamma_1+\gamma_2)^2\gamma_1\gamma_2 + \left(P_S(\gamma_1+\delta)(\gamma_2-\delta)\right)^2\gamma_1\gamma_2.
	\end{align}
	
	Simplifying the above inequality leads to a lower bound on source power budget $P_S$ as: 
	\begin{eqnarray}\label{ps_lower_bound_scp}
	\textstyle P_S > \frac{\sigma^2(\gamma_1+\gamma_2)}{\sqrt{\gamma_1\gamma_2(\gamma_1+\delta)(\gamma_2-\delta)}}.
	\end{eqnarray}
	
	Since, $\gamma_1\gamma_2$ $>$ $(\gamma_1+\delta)(\gamma_2-\delta)$ for all $\delta>0$, the bound can be relaxed as $P_S > \frac{\sigma^2(\gamma_1+\gamma_2)}{\gamma_1\gamma_2}$. Further, it can be observed that $\frac{\gamma_1\gamma_2}{(\gamma_1+\gamma_2)} < \min\{ \gamma_1, \gamma_2\}$. Thus, the condition in \eqref{ps_lower_bound_scp} simply indicates that $\frac{P_S \min\{ \gamma_1, \gamma_2\}}{\sigma^2} > 1$. This condition is generally satisfied because normally SNR is such that $ \frac{P_n^s \gamma_n}{\sigma^2} >1$ over every subcarrier. The condition \eqref{ps_lower_bound_scp} is on $P_S$ the budget. Thus, generally $R>R'$. This proves that rate gets reduced if gap between channel gains is widened. 
	
	It is easy to note that if it is allowed to change channel gains, then the best condition is that all subcarriers should have equal channel gains. Normally, channel gains are random quantities. However, by considering SCP that gives feasibility  of pairing subcarriers on $\mathcal{S}-\mathcal{R}$ and $\mathcal{R}-\mathcal{U}$ link, the possibility of controlling the effective channel gains of all subcarriers appears feasible. Next question is how to match subcarriers to create good effective  channel gains. Ideally the subcarriers should be matched such that effective channel gains on all subcarriers are same. But due to finite number of pairing combinations it may not be possible to have all effective channel gains equal. Thus, the best strategy is to minimize the variance between the effective channel gains.    
	


\begin{thebibliography}{1}
\providecommand{\url}[1]{#1}
\csname url@samestyle\endcsname
\providecommand{\newblock}{\relax}
\providecommand{\bibinfo}[2]{#2}
\providecommand{\BIBentrySTDinterwordspacing}{\spaceskip=0pt\relax}
\providecommand{\BIBentryALTinterwordstretchfactor}{4}
\providecommand{\BIBentryALTinterwordspacing}{\spaceskip=\fontdimen2\font plus
\BIBentryALTinterwordstretchfactor\fontdimen3\font minus
  \fontdimen4\font\relax}
\providecommand{\BIBforeignlanguage}[2]{{%
\expandafter\ifx\csname l@#1\endcsname\relax
\typeout{** WARNING: IEEEtran.bst: No hyphenation pattern has been}%
\typeout{** loaded for the language `#1'. Using the pattern for}%
\typeout{** the default language instead.}%
\else
\language=\csname l@#1\endcsname
\fi
#2}}
\providecommand{\BIBdecl}{\relax}
\BIBdecl



\bibitem{GCW17_SCP}
R.~Saini, D.~Mishra, S.~De, Novel subcarrier pairing strategy for {DF} relayed
  secure {OFDMA} with untrusted users, in: 2017 IEEE Globecom Workshops (GC
  Wkshps), 2017, pp. 1--6.

\bibitem{amitav_TCST_2014}
A.~Mukherjee, S.~Fakoorian, J.~Huang, A.~Swindlehurst, Principles of physical
  layer security in multiuser wireless networks: A survey, IEEE Commun. Surveys
  Tuts. 16~(3) (2014) 1550--1573.

\bibitem{Bassily_SPM_2013}
R.~Bassily, E.~Ekrem, X.~He, E.~Tekin, J.~Xie, M.~Bloch, S.~Ulukus, A.~Yener,
  Cooperative security at the physical layer: A summary of recent advances,
  IEEE Signal Process. Mag. 30~(5) (2013) 16--28.

\bibitem{Wyner1975}
A.~D. Wyner, The wire-tap channel, Bell System Technical Journal 54~(8)
  1355--1387.

\bibitem{Jameel_CST_2018}
F.~Jameel, S.~Wyne, G.~Kaddoum, T.~Q. Duong, A comprehensive survey on
  cooperative relaying and jamming strategies for physical layer security, IEEE
  Communications Surveys Tutorials (2018) 1--1Early access.

\bibitem{LLai_TIT_2008}
L.~Lai, H.~E. Gamal, The relay eavesdropper channel: Cooperation for secrecy,
  IEEE Trans. Inf. Theory 54~(9) (2008) 4005--4019.

\bibitem{Jeong_TSP_2011}
C.~Jeong, I.-M. Kim, Optimal power allocation for secure multicarrier relay
  systems, IEEE Trans. Signal Process. 59~(11) (2011) 5428--5442.

\bibitem{Jindal_CL_2015}
A.~Jindal, R.~Bose, Resource allocation for secure multicarrier {AF} relay
  system under total power constraint, IEEE Commun. Lett. 19~(2) (2015)
  231--234.

\bibitem{Deng_TIFS_2015}
H.~Deng, H.~M. Wang, W.~Guo, W.~Wang, Secrecy transmission with a helper: To
  relay or to jam, IEEE Trans. Inf. Forensics Security 10~(2) (2015) 293--307.

\bibitem{Hmwang_TC_2015}
T.~X. Zheng, H.~M. Wang, F.~Liu, M.~H. Lee, Outage constrained secrecy
  throughput maximization for {DF} relay networks, IEEE Trans. Commun. 63~(5)
  (2015) 1741--1755.

\bibitem{Derrick_TWC_2011}
D.~Ng, E.~Lo, R.~Schober, Secure resource allocation and scheduling for {OFDMA}
  decode-and-forward relay networks, IEEE Trans. Wireless Commun. 10~(10)
  (2011) 3528--3540.

\bibitem{HMWang_TVT_2015}
H.-M. Wang, F.~Liu, M.~Yang, Joint cooperative beamforming, jamming, and power
  allocation to secure {AF} relay systems, IEEE Trans. Veh. Technol. 64~(10)
  (2015) 4893--4898.

\bibitem{Jorswieck2008}
E.~Jorswieck, A.~Wolf, Resource allocation for the wire-tap multi-carrier
  broadcast channel, in: Proc. Int. Conf. Telecommun., Saint-Petersburg,
  Russia, 2008, pp. 1--6.

\bibitem{Xiaowei_TIFS_2011}
X.~Wang, M.~Tao, J.~Mo, Y.~Xu, Power and subcarrier allocation for
  physical-layer security in {OFDMA}-based broadband wireless networks, IEEE
  Trans. Inf. Forensics Security 6~(3) (2011) 693--702.

\bibitem{RSaini_CL_2016}
R.~Saini, D.~Mishra, S.~De, {OFDMA}-based {DF} secure cooperative communication
  with untrusted users, IEEE Commun. Lett. 20~(4) (2016) 716--719.

\bibitem{rsaini_TIFS_2016}
R.~Saini, A.~Jindal, S.~De, Jammer-assisted resource allocation in secure
  {OFDMA} with untrusted users, IEEE Trans. Inf. Forensics Security 11~(5)
  (2016) 1055--1070.

\bibitem{Abedi_TWC_2016}
M.~R. Abedi, N.~Mokari, M.~R. Javan, H.~Yanikomeroglu, Limited rate feedback
  scheme for resource allocation in secure relay-assisted ofdma networks, IEEE
  Transactions on Wireless Communications 15~(4) (2016) 2604--2618.

\bibitem{Abedi_TWC_2017}
M.~R. Abedi, N.~Mokari, H.~Saeedi, H.~Yanikomeroglu, Robust resource allocation
  to enhance physical layer security in systems with full-duplex receivers:
  Active adversary, IEEE Transactions on Wireless Communications 16~(2) (2017)
  885--899.

\bibitem{Poursajadi_TVT_2018}
S.~Poursajadi, M.~H. Madani, Analysis and enhancement of joint security and
  reliability in cooperative networks, IEEE Transactions on Vehicular
  Technology (2018) 1--1Early access.

\bibitem{Zhang_JSAC_2018}
H.~Zhang, N.~Yang, K.~Long, M.~Pan, G.~K. Karagiannidis, V.~C.~M. Leung, Secure
  communications in {NOMA} system: Subcarrier assignment and power allocation,
  IEEE Journal on Selected Areas in Communications 36~(7) (2018) 1441--1452.

\bibitem{RSaini_CL_2017}
R.~Saini, D.~Mishra, S.~De, Utility regions for {DF} relay in {OFDMA}-based
  secure communication with untrusted users, IEEE Communications Letters
  21~(11) (2017) 2512--2515.

\bibitem{Herdin_ICC_2006}
M.~Herdin, A chunk based {OFDM} amplify-and-forward relaying scheme for 4{G}
  mobile radio systems, in: Proc. IEE ICC, Vol.~10, Istanbul, Turkey, 2006, pp.
  4507--4512.

\bibitem{Hottinen_SPAWC_2007}
A.~Hottinen, T.~Heikkinen, Optimal subchannel assignment in a two-hop {OFDM}
  relay, in: Proc. IEEE Wksp. Signal Process. Adv. Wireless Commun., Helsinki,
  Finland, 2007, pp. 1--5.

\bibitem{YLI_CL_2009}
Y.~Li, W.~Wang, J.~Kong, M.~Peng, Subcarrier pairing for amplify-and-forward
  and decode-and-forward {OFDM} relay links, IEEE Commun. Lett. 13~(4) (2009)
  209--211.

\bibitem{HSU_TSP_2011}
C.~Hsu, H.~Su, P.~Lin, Joint subcarrier pairing and power allocation for ofdm
  transmission with decode-and-forward relaying, IEEE Transactions on Signal
  Processing 59~(1) (2011) 399--414.

\bibitem{Guftar_TC_2011}
G.~A.~S. Sidhu, F.~Gao, W.~Chen, A.~Nallanathan, A joint resource allocation
  scheme for multiuser two-way relay networks, IEEE Trans. Commun. 59~(11)
  (2011) 2970--2975.

\bibitem{CCAI_WCSP_2015}
C.~Cai, Y.~Cai, R.~Wang, W.~Yang, W.~Yang, Resource allocation for physical
  layer security in cooperative {OFDM} networks, in: Proc. Intl. Wireless
  Commun. Signal Process. (WCSP), Nanjing, China, 2015, pp. 1--5.

\bibitem{ZChang_ICC_2016}
Z.~Chang, X.~Hou, X.~Guo, T.~Ristaniemi, Energy efficient resource allocation
  for secure {OFDMA} relay systems with eavesdropper, in: Proc. IEEE ICC, Kuala
  Lumpur, Malaysia, 2016, pp. 1--7.

\bibitem{Zhang_TII_2016}
H.~Zhang, H.~Xing, J.~Cheng, A.~Nallanathan, V.~C.~M. Leung, Secure resource
  allocation for {OFDMA} two-way relay wireless sensor networks without and
  with cooperative jamming, IEEE Trans. Ind. Informat. 12~(5) (2016)
  1714--1725.

\bibitem{Baz}
M.~S. Bazaraa, H.~D. Sherali, C.~M. Shetty, Nonlinear Programming: Theory and
  Applications, New York: John Wiley and Sons, 2006.

\bibitem{Y_F_LIU_TSP_2014}
Y.-F. Liu, Y.-H. Dai, On the complexity of joint subcarrier and power
  allocation for multi-user {OFDMA} systems, IEEE Trans. Signal. Process.
  62~(3) (2014) 583--596.

\bibitem{Bookboyd}
S.~Boyd, L.~Vandenberghe, Convex Optimization, Cambridge University Press,
  2004.

\bibitem{W_Yu_TCOM_2006}
W.~Yu, R.~Lui, Dual methods for nonconvex spectrum optimization of multicarrier
  systems, IEEE Trans. Commun. 54~(7) (2006) 1310--1322.



\end{thebibliography}
\end{document}